\documentclass[a4paper,UKenglish,cleveref, autoref, thm-restate]{lipics-v2019}

\usepackage{mathtools}
\usepackage{xspace} 
\usepackage[appendix=append]{apxproof}
\newcommand{\define}[1]{\textit{#1}}
\newcommand{\stress}[1]{\textit{#1}}
\newcommand{\name}[1]{\textsc{#1}\xspace}
\newcommand{\problemm}[1]{\textsc{#1}\xspace}

\newcommand{\Oh}{\mathcal{O}}
\newcommand{\fpt}{\name{FPT}}
\newcommand{\xp}{\name{XP}}
\newcommand{\np}{\name{NP}}
\newcommand{\wone}{\name{W[1]}}
\newcommand{\wtwo}{\name{W[2]}}
\newcommand{\N}{\mathbb{N}}
	
\newcommand{\HHH}{\{\tfrac{z}{2} \mid z \in \mathbb{Z}\}}
\newcommand{\dispersion}{\textsc{Dispersion}\xspace}

\newcommand{\disp}[1]{#1\text{-}\mbox{disp}}
\newcommand{\dispauto}[1]{#1\text{-}\mbox{auto}\text{-}\allowbreak\mbox{disp}}
\newcommand{\tw}{\operatorname{\mathsf{tw}}}
\newcommand{\td}{\operatorname{\mathsf{td}}}
\newcommand{\fvs}{\operatorname{\mathsf{fvs}}}

\newcommand{\pw}{\operatorname{\mathsf{pw}}}

\newcommand{\vc}{\operatorname{\mathsf{vc}}}

\newcommand{\commentt}[1]{}
\newcommand{\refstar}{${(\star)}$}
\newcommand{\backward}{(\( \Leftarrow \))\xspace}
\newcommand{\forward}{(\( \Rightarrow \))\xspace}
\newcommand{\pivots}[2]{{\operatorname{\mathsf{pivots}}}(#1,#2)}
\newcommand{\vel}{\operatorname{\mathsf{vel}}}
\newcommand{\velo}{\operatorname{\mathsf{vel}_S}} 
\newcommand{\sgn}{\operatorname{\mathsf{sgn}}}
\newcommand{\oripara}{\mathsf{dir}}
\newcommand{\noripara}{\overline{\mathsf{dir}}}
\newcommand{\ori}[2]{\oripara({#2}\to{#1})}

\newcommand{\nori}[2]{\overline{\oripara}({#2}\to{#1})}

\newcommand{\lfs}{\operatorname{\mathsf{pos}}} 
\newcommand{\lf}[2]{\lfs_{#1}(#2)} 
\newcommand{\fracp}{\operatorname{\mathsf{fp}}} 
\newcommand{\change}{\operatorname{\mathsf{flip}}}
\newcommand{\orie}[2]{\mathsf{dir}^\star({#2}\to{#1})}
\newcommand{\norie}[2]{\overline{\mathsf{dir}}^\star({#2}\to{#1})}

\makeatletter
\newcommand{\labeltext}[3][]{%
    \@bsphack%
    \csname phantomsection\endcsname
    \def\tst{#1}%
    \def\labelmarkup{}
    \def\refmarkup{}%
    \ifx\tst\empty\def\@currentlabel{\refmarkup{#2}}{\label{#3}}%
    \else\def\@currentlabel{\refmarkup{#1}}{\label{#3}}\fi%
    \@esphack%
    \labelmarkup{#2}
}
\makeatother

\usepackage{tikz}
\usetikzlibrary{calc}

\newcommand{\pointe}[4]{
	\node [point] (internal) at ($ (#2) !{#4}! (#3) $) {};
	\node [right] at (internal.east) {$#1$};
}
\newcommand{\points}[4]{
	\node [point] (internal) at ($ (#2) !{#4}! (#3) $) {};
	\node [below] at (internal.south) {$#1$};
}

\newcommand{\pointn}[4]{
	\node [point] (internal) at ($ (#2) !{#4}! (#3) $) {};
	\node [above] at (internal.north) {$#1$};
}

\title{Dispersing Obnoxious Facilities on Graphs
by
Rounding Distances}

\titlerunning{Dispersing Obnoxious Facilities on Graphs
by
Rounding Distances}

\author{Tim A. Hartmann}{Department of Computer Science, RWTH Aachen University, Germany}{hartmann@algo.rwth-aachen.de}{https://orcid.org/0000-0002-1028-6351}{}

\author{Stefan Lendl}{Department of Operations and Information Systems, University of Graz, Austria}{}{https://orcid.org/0000-0002-5660-5397}{}

\authorrunning{T.\,A. Hartmann and S. Lendl}

\Copyright{Tim A. Hartmann and Stefan Lendl} 

\ccsdesc[500]{Theory of computation~Discrete optimization}
\ccsdesc[500]{Mathematics of computing~Graph algorithms}
\ccsdesc[500]{Theory of computation~Parameterized complexity and exact algorithms}
\ccsdesc[500]{Theory of computation~Problems, reductions and completeness}

\keywords{facility location, parameterized complexity, packing}

\category{} 

\relatedversion{} 

\supplement{}

\nolinenumbers 

\EventEditors{}
\EventNoEds{1}
\EventLongTitle{}
\EventShortTitle{}
\EventAcronym{}
\EventYear{}
\EventDate{}
\EventLocation{}
\EventLogo{}
\SeriesVolume{}
    \ArticleNo{1}

\begin{document}

\maketitle

\begin{abstract}
We continue the study of $\delta$-dispersion, a continuous facility location problem on a graph where all edges have unit length and where the facilities may also be positioned in the interior of the edges.
The goal is to position as many facilities as possible subject to the condition that every two facilities have distance at least $\delta$ from each other.

Our main technical contribution is an efficient procedure to `round-up' distance $\delta$. It transforms a $\delta$-dispersed set $S$ into a $\delta^\star$-dispersed set $S^\star$ of same size
	where distance $\delta^\star$ is a slightly larger rational $\tfrac{a}{b}$ with a numerator $a$ upper bounded by the longest (not-induced) path in the input graph.

Based on this rounding procedure and connections to the distance-$d$ independent set problem we derive a number of algorithmic results. When parameterized by treewidth, the problem is in XP. When parameterized by treedepth the problem is FPT and has a matching lower bound on its time complexity under ETH. Moreover, we can also settle the parameterized complexity with the solution size as parameter using our rounding technique: $\delta$-\dispersion is FPT for every $\delta \leq 2$ and W[1]-hard for every $\delta > 2$.

Further, we show that $\delta$-dispersion is NP-complete for every fixed irrational distance $\delta$, which was left open in a previous work.	
\end{abstract}

\section{Introduction}

We study the algorithmic behavior of a continuous dispersion problem.
Consider an undirected graph $G$, whose edges are have unit length.
Let $P(G)$ be the continuum set of points on all the edges and vertices.
For two points $p,q \in P(G)$, we denote by $d(p,q)$ the length of a shortest path containing $p$ and $q$ in the underlying metric space.
A subset $S \subseteq P(G)$ is \define{$\delta$-dispersed} for some positive real number $\delta$,
	if every distinct points $p,q \in S$ have distance at least $d(p,q) \geq \delta$.
Our goal is,
	for a given graph $G$ and a positive real number $\delta$,
	to compute a maximum cardinality subset $S \subseteq P(G)$ that is $\delta$-dispersed.
We denote by \define{$\disp{\delta}(G)$} the maximum size of a $\delta$-dispersed set of~$G$.
The decision problem \dispersion asks for a $\delta$-dispersed set of size at least $k$, 
	where additionally integer $k\geq 0$ is part of the input.
When $\delta$ is fixed and not part of the input, we refer to the problem as $\delta$-\dispersion.

\subsection{Known and related results}

The area of obnoxious facility location goes back to seminal articles of Goldman \& Dearing~\cite{GolDea1975} and Church \& Garfinkel~\cite{ChuGar1978}.
The area includes a wide variety of objectives and models.
For example, purely geometric variants have been studied by Abravaya \& Segal~\cite{AbrSeg2010}, Ben-Moshe, Katz \& Segal~\cite{BenKatSeg2000}, and Katz, Kedem \& Segal~\cite{KatKedSeg2002}.
Recently, van Ee studied the approximability of a generalized covering problem in a metric space that also involves dispersion constraints~\cite{DBLP:journals/dam/Ee21}.
Another direction is a graph-theoretic model, where every edge of the given graph $G$ is rectifiable and has some individual length.
Tamir discusses the complexity and approximability of several optimization problems.
For example, when $G$ is a tree, then a $\delta$-dispersed set can be computed in polynomial time \cite{Tamir1991}.
Another task is to place a single obnoxious facility in a network while maximizing, for example, the smallest distance from facility to certain clients, as studied by Segal~\cite{Segal2003}.

In a previous work, the complexity of \dispersion was studied for every rational distance~$\delta$.
When $\delta$ is a rational number with numerator $1$ or $2$,
	the problem is polynomial time solvable,
	while it is \np-complete for all other rational values of $\delta$~\cite{DBLP:conf/stacs/GrigorievHLW19,grigoriev2021dispersing}.
The complexity when $\delta$ is irrational was left as an open problem.

A closely related facility location problem 
is $\delta$-covering.
The objective is to place as few locations as possible on $P(G)$ subject to the condition that any point in $P(G)$ is in distance at most $\delta$ to a placed location.
This problem is polynomial time solvable whenever $\delta$ is a unit fraction,
	while it is \np-hard for all non unit fractions~$\delta$~\cite{DBLP:journals/mp/HartmannLW22}.
	Furthermore, the parameterized complexity with the parameter solution size $k$ is studied.
	$\delta$-covering is fixed parameter tractable when $\delta < \tfrac{3}{2}$,
	while for $\delta \geq \tfrac{3}{2}$ the problem is \wtwo-complete~\cite{DBLP:journals/mp/HartmannLW22}.
Tamir~\cite{DBLP:journals/mor/Tamir87} showed that for $\delta$-covering only certain distances $\delta$ are of interest.
For every amount of points $p$ the distance $\max\{ \delta^\star : |\mbox{cover}{\delta^\star}(G)|=p \}$ is of the form $\tfrac{L'}{2p'}$ where $p' \in \{1,\dots,p\}$ and $L'$ is roughly at most twice the length of a non-induced path in $G$.

\subsection{Our contribution}
Our main technical contribution is an efficient and constructive rounding procedure.
Given a $\delta$-dispersed set $S$ for some distance value $\delta > 0$,
	it transforms $S$ into a $\delta^\star$-dispersed set $S^\star$ of equal size
	with a slightly larger well-behaving distance value $\delta^\star \geq \delta$.
The new distance $\delta^\star$ is a rational $\tfrac{a}{b}$ with small numerator $a$.
More precisely, the numerator is upper bounded by the length of the longest (not-induced) path $L$, 
hence upper bounded asymptotically by the number of vertices $n$ of the input graph (see \autoref{sec:rounding}).

Our second technical contribution relates the optimal solution for distance $\delta$ and $\tfrac{\delta}{\delta+1}$ for $\delta \leq 3$.
A $\delta$-dispersed set translates to a $\tfrac{\delta}{\delta+1}$-dispersed set
	by placing one more point on every edge, and vice versa by removing one point (see~\autoref{sec:translating}).

Further we explore a connection of \dispersion and an independent set problem (see~\autoref{section:relationship:with:dis}).
The combination of that connection with our technical contributions yields several algorithmic results for \dispersion (see~\autoref{section:algo:imp} and~\autoref{sec:hardness}):
\begin{itemize}
	\item
	\dispersion is \np-hard even for chordal graphs of diameter $4$.
	\item
	\dispersion is \fpt for the graph parameter treedepth $\td(G)$
		with a run time matching a lower bound under ETH.
	We complement this result by showing that $\delta$-\dispersion is \wone-hard
		for the slightly more general graph parameter pathwidth $\pw(G)$, even for the combined parameter $\pw(G) + k$.
	Similarly, $\delta$-\dispersion is \wone-hard for the graph parameter $\fvs(G)$, the minimum size of a feedback vertex set.
	\item
	\sloppy
	\dispersion is \xp for the parameter treewidth $\tw(G)$, with a running time of $(2L+2)^{\tw(G)}n^{\Oh(1)}$,
		where $n$ is the number of vertices and $L$ is an upper bound on the length of the longest path in $G$.
	We complement this result by the more general lower bound of $n^{o(\tw(G)+ \sqrt{k})}$, assuming ETH.
	It implies the lower bound of $L^{o(\tw(G)+ \sqrt{k})}$ since $L \leq n$.
	Note that a mere lower bound of $L^{o(\tw(G)+ \sqrt{k})}$ would not exclude an $n^{o(\tw(G))}$-algorithm.
\end{itemize}

In addition, we completely resolve the complexity of $\delta$-dispersion, by showing \np-hardness for irrational $\delta$ (see~\autoref{sec:irrational}).
We also study the parameterized complexity when parameterized by the solution size $k$.
The problem is \wone-hard when $\delta >2$, and \fpt otherwise.
Thus, there is a sharp threshold at $\delta = 2$ where the complexity jumps from \fpt to \wone-hard (see~\autoref{sec:naturalparameter}).

	We mark statements whose proof can be found in the appendix with ``\refstar''.
	
\section{Preliminaries}

We use the word \define{vertex} in the graph-theoretic sense,
	while we use the word \define{point} to denote the elements of the geometric structure $P(G)$.
As an input for $\delta$-dispersion,
	we consider graphs $G$ that are undirected, connected, and without loops and isolated vertices.

For an edge $\{u,v\} \in E(G)$ and a real number $\lambda \in [0,1]$,
	let $p(u,v,\lambda) \in P(G)$ be the point on edge $\{u,v\}$ that has distance $\lambda$ from $u$.
Note that $p(u,v,0)=u$, $p(u,v,1)=v$ and $p(u,v,\lambda) = p(v,u,1-\lambda)$.
Further, we use $d(p,q)$ for the length of a shortest path between points $p,q \in P(G)$.

For a subset of vertices $V' \subseteq V(G)$ or a subset of edges $E' \subseteq E(G)$, we denote by $G[V']$ and $G[E']$ the subgraph induced by $V'$ and $E'$, respectively.
The neighborhood of a vertex $u$ is $N(u) \coloneqq \{ v \in V(G) \mid \{u,v\}\in E(G)\}$.
We use $n$ as the number of vertices of $G$, when $G$ is clear from the context.

For a graph $G$ and integer $c\geq 1$,
	let the \define{$c$-subdivision of $G$} be the graph $G$
	where every edge is replaced by a path of length $c$.

\begin{lemma}[\cite{grigoriev2021dispersing}]
	\label{lemma:subdivide}
	Let \( G \) be a graph, let \( c \geq 1 \) be an integer, and let \( G' \) be the \( c \)-subdivision of~\( G \).
	Then \( \disp{\delta}(G) = \disp{(c \delta)}(G') \).
\end{lemma}

For integers $a$ and~$b$, we denote the rational number $\frac{a}{b}$ as \define{$b$-simple}.
A set $S\subseteq P(G)$ is $b$-simple,
	if for every point $p(u,v,\lambda)$ in $S$ the edge position $\lambda$ is \define{$b$-simple}.

\begin{lemma}[\cite{grigoriev2021dispersing}]
\label{lemma:simple:solution}
	Let \( \delta = \tfrac{a}{b} \) with integers \( a \) and \( b \), and let $G$ be a graph.
	Then, there exists an optimal \( \delta \)-dispersed set \( S^\star \) that is \( 2b \)-simple.
\end{lemma}

For an introduction into parameterized algorithms, we refer to \cite{DBLP:books/sp/CyganFKLMPPS15}.
We study of the complexity of \dispersion with the natural parameter solution size $k$,
	as well as its dependency on structural measures on the input graph.
Besides treewidth $\tw(G)$ and pathwidth $\pw(G)$,
	we also study the parameters `feedback vertex set size' $\fvs(G)$ and treedepth $\td(G)$.

A graph has a feedback vertex set $W \subseteq V(G)$ if $G$ after removing $W$ contains no cycle.
The `feedback vertex set size' is the size of a smallest feedback vertex set of $G$.

The treedepth of a connected graph $G$ can be defined as follows.
If $G$ is disconnected, it is the maximum treedepth of its components;
If $G$ consists of a single vertex, then $\td(G)=1$;
And else it is one plus the minimum over all $u \in V(G)$ of the treedepth of $G$ without vertex~$u$.

We provide lower bounds for the time-complexity assuming the Exponential Time Hypothesis (ETH):
There is no $2^{o(N)}$-time algorithm for $3$-\problemm{SAT} with $N$ variables and $\Oh(N)$ clauses~{\cite{DBLP:journals/jcss/ImpagliazzoPZ01}.  
For more details on ETH, we refer to~\cite{DBLP:books/sp/CyganFKLMPPS15}.

\section{Translating $\delta$-Dispersion}
\label{sec:translating}

There is an intriguing relation of the optimal solution for distance $\delta$ and $\tfrac{\delta}{2\delta+1}$
	for the similar problem $\delta$-covering~\cite{DBLP:journals/mp/HartmannLW22}.
We may analogously expect that an optimal solution for $\delta$-dispersion translates to an optimal solution for $\tfrac{\delta}{\delta+1}$-dispersion;
	i.e., that an optimal $\delta$-dispersed set corresponds to an optimal $\tfrac{\delta}{\delta+1}$-dispersed set of the same size plus one extra point for every edge.

This is not true for $\delta = 3+\varepsilon$ for any $\varepsilon > 0$:
Consider a triangle, where a $(3+\varepsilon)$-dispersed set $S$ contains at most one point $p$.
Since $\frac{\delta}{\delta + 1} > \tfrac{3}{4}$, a $\frac{\delta}{\delta + 1}$-dispersed set however contains at most $3 < |S|+3$ points.

Causing trouble is a non-trivial closed walk containing $p$ and of length less than $\delta$.
The translating lemma may only apply to a variation of dispersion that is sensitive to such walks,
	a variant which we call \define{auto-dispersion}.
A $\delta$-dispersed set $S \subseteq P(G)$ is \define{$\delta$-auto-dispersed}
	if additionally for every point $p \in S$
	there is no walk from $p$ to $p$ of length $<\delta$ that is locally-injective and non-trivial.
A walk is \define{locally-injective} if, when interpreted as a continuous mapping $f: [0,1] \to P(G)$ from $f(0)=p$ to $f(1)=p$,
	has for every pre-image $c \in (0,1)$ a positive range $\varepsilon > 0$ such that $f$ restricted to the interval $(c-\varepsilon, c+ \varepsilon)$ is injective.

\newcommand{\lemmaTextLimitBAuto}{
	Let \( G \) be a graph and $ \delta > 0$.
	Then $ \dispauto{\delta}(G) = \dispauto{\frac{\delta}{\delta+1}}(G) + |E(G)| .$}

\begin{lemma}
\label{lemma:limit:b:auto}
\lemmaTextLimitBAuto
\end{lemma}

Fortunately, this translation lemma is still useful for ordinary $\delta$-dispersion.
We have $\dispauto{\delta}(G) = \disp{\delta}(G)$ for $\delta \leq 3$,
	since there is no such non-trivial closed walk 'without a turn within an edge' of length $<3$.
The threshold of $3$ is tight according to the above example with graph~$K_3$.

\begin{corollary}\label{lemma:limit:b}
Let $G$ be a graph and $\delta \in (0,3]$.
Then 
$
\disp{\delta}(G) = \disp{\frac{\delta}{\delta+1}}(G) + |E(G)| .
$
\end{corollary}

\section{Dispersion and Independent Set}
\label{section:relationship:with:dis}

\newcommand{\problemdefSimple}[3]{
\medskip
\begin{tabularx}{\textwidth}{ r X p{0.5cm} }
\multicolumn{2}{l}{#1} \\
Input: & #2 \\	
Question: & #3
\end{tabularx}
}

\newcommand{\ddis}[1]{\problemm{Distance $#1$ Independent Set}}
\newcommand{\dis}{\problemm{Distance Independent Set}}
\newcommand{\ddisdef}{
\problemdefSimple{\ddis{d}}{
A graph $G$ and an integer $k$.}{
Is there a subset $I \subseteq V(G)$ of size $|I|\geq k$
	such that each distinct $u,v \in I$ have distance at least $d$?
}
}

To solve \dispersion we can borrow from algorithmic results from a
	generalized independent set problem.
A classical independent set is a set of \stress{vertices}
	where each two elements have to be at least~$2$ apart from each other (when we consider that the edges have unit length).
In a $2$-dispersed set also each two elements need to be at least~$2$ apart from each other,
	though the set contains a set of \stress{points} of the graph.

To generalize the independent set problem,
	we may ask that the vertices are not~$2$ apart but some integer~$d$ apart from each other.
Such a generalization for independent set
	is called a \define{distance-$d$ independent} set or \define{$d$-scattered set}.
They have been studied by Eto~et~al.~\cite{DBLP:journals/jco/EtoGM14} and Katsikarelis et al.~\cite{DBLP:journals/dam/KatsikarelisLP22}.

Let $\alpha_d(G)$ be the maximum size of a distance $d$ independent set, for a graph~$G$ and integer~$d$.
We relate $\delta$-dispersion to
	$\alpha_d$. 
We consider the \define{$c$-subdivision of a graph $G$},
	denoted as $G_c$,
	which is the graph~$G$ where every edge is replaced by a path of length $c$,
	for some integer~$c\geq 1$.

\newcommand{\lemmaTextSubdivisionIntoIS}{
Consider integers $a,b$ and a $2b$-subdivision $G_{2b}$ of a graph $G$.
Then $\disp{\tfrac{a}{b}}(G) = \alpha_{2a}(G_{2b})$.
}

\begin{lemma}
\label{lemma:subdivision:into:is}
\lemmaTextSubdivisionIntoIS
\end{lemma}
\begin{proof} 
Consider the $b$-subdivision $G_b$ of $G$.
Then $G_{2b}$ is a $2$-subdivision of~$G_b$.
We know that $\disp{\tfrac{a}{b}}(G) = \disp{2a}(G_{2b})$ from~\cref{lemma:subdivide}.
Hence it remains to show $\disp{2a}(G_{2b}) = \alpha_{2a}(G_{2b})$.

Clearly, a distance-$2a$ independent set $I \subseteq V(G)$ is also a $2a$-dispersed set.
For the reverse direction,
	assume there is a $2a$-dispersed set $I_{2a}$ of $G_{2b}$.
Then $I_{2a}$ corresponds to an $a$-dispersed set $I$ of $G_b$ of same size, according to \cref{lemma:subdivide}.
Since $a$ is integer, we may assume that $S$ contains only half-integral points,
	hence points with edge position from $\{0,\tfrac{1}{2},1\}$, according to~\autoref{lemma:simple:solution}.
Let $G_{2b}$ result from $G_b$ by replacing each edge $\{u,v\}$ by a path $u w_{u,v} v$.
Then let $I \subseteq V(G_{2b})$ consist of vertex $u \in V(G)$ with a point in $S$
	and every $w_{u,v}$ for every point $p(u,v,\tfrac{1}{2}) \in S$.
Then $I$ is a distance $2a$ independent set of $G_{2b}$ of size ${|I|=|S|}$.	
\end{proof}

Thus to solve \dispersion for $\delta = \tfrac{a}{b}$ we can use algorithms for distance $d$ independent set.
For rationals $\tfrac{a}{b}$ with small values of $a$ and~$b$ this possibly leads to efficient algorithms.
For example, a distance $d$ independent set on graphs width treewidth $\tw(G)$ (and a given tree decomposition) can be found in time $d^{\tw(G)} n^{\Oh(1)} $, see \cite{DBLP:journals/dam/KatsikarelisLP22}.
``Simply'' subdivide the edges of the input graph sufficiently often,
	which does not increase the treewidth of the considered graph.
To find a $\tfrac{a}{b}$-dispersed set in a graph~$G$, we can search for distance $2a$ independent set the $2b$-subdivision of~$G$.

\begin{corollary}
\label{lemma:naive:tw:algo}
There is an algorithm that, given a rational distance $\tfrac{a}{b} > 0$ and a graph \( G \), a tree decomposition of width $\tw(G)$,
	computes a maximum \( \tfrac{a}{b}  \)-dispersed set \( S \) in time $(2a)^{\tw(G)} (bn)^{\Oh(1)}$.
\end{corollary}

However, in general this constitutes a possibly exponential increase of the input size.
While in the input of $\tfrac{a}{b}$-dispersion encodes $a$ and $b$ in binary,
	the subdivided graph essentially encodes~$b$ in unary. 
Further, if $\delta$ is irrational, we do not have a suitable subdivision at all.

\section{Rounding the Distance}
\label{sec:rounding}

For a given graph $G$ and distance $\delta$,
	we state a rational $\delta^\star \geq \delta$
	such that $\disp{\delta}(G) = \disp{\delta^\star}(G)$.
Our proof is constructive.
We give a procedure that efficiently transforms a $\delta$-dispersed set into a $\delta^\star$-dispersed set.
The guaranteed rational $\delta^\star$ has a numerator bounded by the longest path in $G$ (or just $n$ as an upper bound thereof).
It is independent of the precise structure of the given graph.

To give some intuition:
Generally there is some leeway for $\delta$.
For example, in a star $K_{1,k}$, $k\geq 1$ for every $\delta \in (1,2]$ the optimal solution puts a point on every leaf yielding a $\delta$-dispersed set of size $k$.
Hence for instance $\disp{\tfrac{3}{2}}(K_{1,k}) = \disp{2}(K_{1,k})$.
However, for $\delta>2$ only one point can be placed,
	such that $\disp{2}(K_{1,k}) \neq \disp{(2+\varepsilon)}(K_{1,k})$ for every $\varepsilon>0$.
	
So what $\delta^\star$ can be guaranteed such that $\disp{\delta}(G) = \disp{\delta^\star}(G)$?
An illustrative example is a path of length $6$.
Then $\disp{\tfrac{15}{11}}(G) = 5 = \disp{\tfrac{3}{2}}(G)$.
For $\delta = \tfrac{15}{11}$ tightly packing $5$ points allows to have a space of size $\tfrac{6}{11}$ at either end of the path, not enough to place another point.
However, placing $5$ points in distance $\delta=\tfrac{3}{2}$ allows no leeway;
	$\delta$ is (already) a divisor of $6$, the length of the considered path.
Distance $\delta^\star$ relies on $L$, the length of the longest (not-induced) path in $G$.
We have to take into account that $\delta$ might divide any path of length $\leq L$.
Our $\delta^\star$ is the smallest rational $\tfrac{a^\star}{b^\star}$ where the numerator $a^\star \leq 2L+2$.
In other words, the inverse of $\delta^\star$ is the next rational number in the Farey sequence of order $2L+2$.

\begin{theorem}
	\label{lemma:rounding}
	Let \( \delta \in \mathbb{R}^{+} \).
	Let $L$ be an upper bound on the length of paths in $G$.
	Let \( \delta^\star = \frac{a^\star}{b^\star} \geq \delta \) minimal for $a^\star \leq 2L+2$ and \( b^\star \in \mathbb{N} \).
	Then \( \disp{\delta}(G) = \disp{ \delta^\star }(G) \).
\end{theorem}

Clearly, a \( \delta^\star \)-dispersed set $S^\star$, is also \( \delta \)-dispersed, since $\delta^\star \geq \delta$.
We have to show the reverse direction.
Consider a \( \delta \)-dispersed set $S$ (of size $|S|\geq 2$) of a connected graph $G$ that is not \( \delta^\star \)-dispersed, hence $\delta$ is irrational or is equal to $\frac{a}{b}$ for some co-prime $a,b$ with \( a > 2L+2 \).

\medskip

In the following we develop our rounding procedure that shows the reverse direction.
Our presentation aims to be accessible by starting from the core algorithmic idea
	from which we unravel all involved technical concepts piece by piece.
The detailed proofs are placed in the appendix.

\subsection{Overview}

Our rounding procedure repeatedly applies a pushing algorithm to the current point set $S$.
We show that each such step strictly decreases a polynomially bounded potential $\Phi: P(G) \to \N$.

\begin{theorem}
\label{lemma:overview}
Suppose that there is an algorithm, that given a $\delta$-dispersed set $S$ with $\delta < \delta^\star$
	computes an $\varepsilon > 0$ and a $(\delta + \varepsilon)$-dispersed set $S_{\varepsilon}$ of size $|S_\varepsilon| = |S|$
	that satisfies $\Phi(S) > \Phi(S_\varepsilon)$
	for some polynomially bounded potential $\Phi: P(G) \to \N$.
Then \autoref{lemma:rounding} follows.
\end{theorem}
\begin{proof}
Let $S$ be a $\delta$-dispersed set.
Apply the assumed algorithm to obtain a $\varepsilon > 0$ and a $(\delta + \varepsilon)$-dispersed set $S_{\varepsilon}$ of size $|S_\varepsilon| = |S|$.
If $\delta = \delta^\star$, we reached our goal.
Else we apply the assumed algorithm again.
Since the potential $\Phi: P(G) \to \N$ decreases for $S_{\varepsilon}$ compared to $S$
	and $\Phi$ is polynomially bounded,
	we have to reach $\delta^\star$ in polynomial many steps.
\end{proof}

In the remainder of this section we will develop such an algorithm.
It pushes the points of point set $S$ away from each other
	such that their pairwise distance increases from `at least $\delta$' to `at least $\delta + \varepsilon$'.
We choose $\varepsilon \geq 0$ as a large as possible limited by some events.
Either we already reach $\delta + \varepsilon = \delta^\star$, hence we reached our goal,
	or at least one of three events occurs.
We will specify these events in the course of this section.
These events mean that one pushing step, i.e., one step for \autoref{lemma:overview} terminated.
All the following preparations for such a pushing step start anew. 

We make sure that our potential $\Phi: P(G) \to \N$ decreases when an event occurs.
Each of the three events has a corresponding partial potential $\Phi_1(S)$, $\Phi_2(S)$ and $\Phi_3(S)$.
They define the overall potential as $\Phi(S) \coloneqq \Phi_1(S) + \Phi_2(S) + \Phi_3(S)$.
Each part never increases.
Whenever event $i$ occurs, $\Phi_i(S)$ strictly decreases.

We denote a pair of points $\{p,q\}$ from our given point set $S$ as $\delta$-\define{critical},
	if they have distance exactly $\delta$.
Hence the critical pairs of points are exactly those that we need to push away from each other.
At the same time we make sure that, once $\{p,q\}$ are $\delta$-critical,
	they never turn uncritical again, i.e., they are $(\delta+\varepsilon)$-critical in the next step.
An uncritical pair of points $\{p,q\}$ might become critical,
	hence we have to take care of $\{p,q\}$ in future steps.
This constitutes our first event.
The corresponding partial potential is $\Phi_1(S)$,
	the number of uncritical pairs of points $\{p,q\}$.
\begin{quote}
(\labeltext{Event 1}{event:1:n})
	A $\delta$-uncritical pair of points $\{p,q\}$ becomes $(\delta+\varepsilon)$-critical.
\end{quote}
$$
\Phi_1(S) \;\coloneqq\; \big|\big\{ \{p,q\} \in \textstyle\binom{S}{2} \mid \{p,q\} \text{ are not $\delta$-critical}  \big\}\big|
	\;\leq\; |S|^2
$$

\begin{figure}[t]
\begin{center}
\begin{subfigure}[b]{0.25\textwidth}
\begin{tikzpicture}
[scale=1.5,auto=left, node/.style={rectangle,fill=white, draw, scale=0.5},
	point/.style={circle,fill=black, scale=0.5}]
	
	\node [node] (ll) at (0,0) {};
	
	\node [node] (l) at (1,0) {};
	\node [node] (t) at (1.71,{sin 30}) {};
	\node [right] at (t.east) {$v$};
	\node [node] (b) at (1.71,-{sin 30}) {};
	\node [right] at (b.east) {$u$};

	\points{p}{l}{ll}{0.6}	
	\pointe{q}{b}{t}{0.4}
	
	\foreach \from/\to in {ll/l,l/t,l/b,t/b}
	\draw (\from) -- (\to);
\end{tikzpicture}
\end{subfigure}
\hfill
\begin{subfigure}[b]{0.32\textwidth}
\begin{tikzpicture}
[scale=1.5,auto=left, node/.style={rectangle,fill=white, draw, scale=0.5},
	point/.style={circle,fill=black, scale=0.5}]
	
	\node [node] (ll) at (0,0) {};
	
	\node [node] (l) at (1,0) {};
	\node [node] (t) at (1.71,{sin 30}) {};
	\node [node] (b) at (1.71,-{sin 30}) {};
	\node [node] (r) at (2.42,0) {};

	\points{p_0}{ll}{l}{0.55}
	\points{p_1}{l}{b}{0.70}	
	\points{p_2}{b}{r}{0.85}
	\pointn{p_3}{l}{t}{0.90}

	\foreach \from/\to in {ll/l,l/t,l/b,t/r,r/b}
	\draw (\from) -- (\to);
\end{tikzpicture}
\end{subfigure}
\hfill
\begin{subfigure}[b]{0.39\textwidth}
\begin{tikzpicture}
[scale=1.5,auto=left, node/.style={rectangle,fill=white, draw, scale=0.5},
	point/.style={circle,fill=black, scale=0.5}]
	
	\node [node] (ll) at (0,0) {};
	
	\node [node] (l) at (1,0) {};
	\node [node] (t) at (1.71,{sin 30}) {};
	\node [node] (tt) at (2.71,0.5) {};
	\node [above] at (tt.north) {$v$};
	\node [node] (b) at (1.71,-{sin 30}) {};
	\node [node] (bb) at (2.71,-0.5) {};
	\node [below] at (bb.south) {$u$};

	\points{p_0}{ll}{l}{0.7}
	\points{p_1}{l}{b}{0.3}
	\points{p_2}{l}{b}{0.9}
	\points{p_3}{b}{bb}{0.5}
	\pointe{p_4}{bb}{tt}{0.1}
	\pointe{p_5}{bb}{tt}{0.7}
	\pointn{p_6}{tt}{t}{0.3}
	\pointn{p_7}{tt}{t}{0.9}
	\pointn{p_8}{t}{l}{0.5}

	\foreach \from/\to in {ll/l,l/t,l/b,t/tt,tt/bb,bb/b}
	\draw (\from) -- (\to);
\end{tikzpicture}
\end{subfigure}
\end{center}
\caption{
(left)
Consider critical points $\{p,q\}$ (points depicted as black dots; vertices as white squares).
If we move $q$ away from $p$ by $\varepsilon \geq 0$,
their distance increases by $\varepsilon$ until $q$ reaches the half-integral point $p(u,v,\tfrac{1}{2})$.
(middle)
Let $\{p_{i},p_{i-1}\}$ be critical for $i \geq 1$.
Consider moving $p_i,i\geq 1$ by $i\varepsilon$ away from $p_{i-1}$.
Once $p_2$ becomes half-integral, points $\{p_3,p_0\}$ become also critical, hence
we cannot continue to move points in the same way.
This happens when a point in $S$ becomes half-integral or $\dots$
(right)
$\dots$ a point half-way between two points in $S$ becomes half-integral,
as in this example between $p_4,p_5$.
We say $p_4,p_5$ \define{witness} the \define{pivot} $p(u,v,\tfrac{1}{2})$.}
\label{figure:start}
\end{figure}
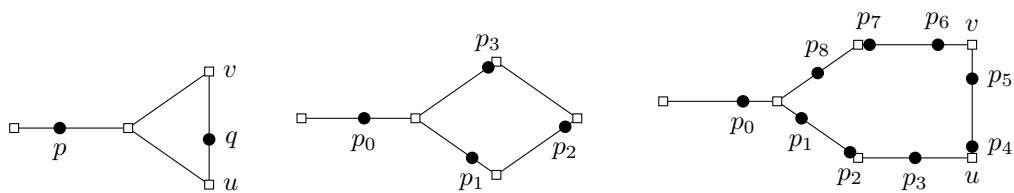

\subsection{Coordination of Movement}

We need to coordinate the movement of all critical pairs of points.
To this end, we will fix some set of \define{root points} $R$.
Our movement will be locally prescribed for sequences of points $p_0,p_1,\dots,p_s$
	that originate in $p_0 \in R$ and where each $\{p_0,p_1\}, \dots ,\{p_{s-1},p_s\}$ is critical.
The overall movement will be uniquely defined by movement defined for these sequences.

For now, consider such a sequence of points $p_0,p_1,p_2,\dots$.
Our idea is to do not move $p_0$,
	to move $p_1$ by distance $\varepsilon$ away from $p_0$,
	point $p_2$ by distance~$2\varepsilon$ away from $p_1$ and so on.
We have to stop pushing in this way
	as soon as one of the points, say $p_i$, becomes half-integral,
	i.e., $p_i$ is moved onto a vertex or the midpoint of an edge.
See \autoref{figure:start} for examples.
This constitutes the second event.
\begin{quote}
(\labeltext{Event 2}{event:2:n})
A non-half-integral $p \in S$ becomes half-integral.
\end{quote}
$$
\Phi_2(S) \;\coloneqq\; \big|\big\{ p \in S \mid p \text{ is not half-integral}\big\}\big| 
\;\leq\; |S|.
$$

The next pushing step will choose $p_i$ as one of the root point $R$ and will move the points away from $p_i$ instead of $p_0$.
Very similarly, we stop when a point $r \in P(G)$ that is `half-way' between two points $p_i,p_{i-1}$ becomes half-integral.
Formally, we denote such a point $r$ as an $(S,\delta)$-\define{pivot}, or simply a pivot,
	if it is half-integral and there is a (critical) pair of points $\{p,q\} \in \binom{S}{2}$,
	the witnesses,
	that have equal distances to $r$,
		which meas $d(p,r) = d(q,r) = \tfrac{\delta}{2}$.
Let $\pivots{S}{\delta}$ be the set of $(S,\delta)$-pivots,
		and let $W(S,\delta) \subseteq \binom{S}{2}$ be the family of pairs of points from $S$, that witness some $(S,\delta)$-pivot.
This leads to the third and final event. 
\begin{quote}
(\labeltext{Event 3}{event:3:n})
A non-pivot point $r \in P(G)$ becomes a pivot.
\end{quote}
$$ \Phi_3(S) \;\coloneqq\; \big|\big\{  r \in P(G) \mid r \text{ is half-integral} \big\} \setminus \pivots{S}{\delta}\big| \;\leq\; |V(G)|^2 . $$

Hence a root point $R$ may not only be a point $p \in S$ but also come from the set of pivots.
We will later properly define $R$ as a superset of half-integral points $p_i \in S$ and the $(S,\delta)$-pivots.

We use an \define{auxiliary graph} $G_S$ for the current $\delta$-dispersed set $S$.
Its vertex set is $S \cup \pivots{S}{\delta}$.
Essentially we make all pairs of critical $\{p,q\}$ adjacent unless they witness a pivot;
If they do witness a pivot, we make them adjacent to the pivot:
\begin{itemize}
	\item
	For $\{p,q\} \in W(S,\delta)$ and for every pivot $r \in \pivots{S}{\delta}$ they witness,
		add edges $\{p,r\},\{r,q\}$; and
	\item
	for every critical pair of points $\{p,q\} \in \binom{S}{2} \setminus W(S,\delta)$ add edge $\{p,q\}$.
\end{itemize}
Note that, for every edge $\{r,p\}$ with $r \in \pivots{S}{\delta}$,
	there is at least one other edge $\{r,q\}$ such that $p,q$ witness $r$ as a pivot.

Now we define the sequence of points which serve as the structure to state the movement.
A path $P=(p_0,p_1,\dots,p_s)$ in the auxiliary graph $G_S$ of length $s\geq 1$  is a \define{spine} if $p_1,\dots,p_{s}$ are not half-integral.
Note that any sub-sequence $(p_0,\dots,p_i)$ for $1\leq i \leq s$ is also a spine.

\subsection{Velocities}

We assign velocities $\vel_P$ to the points $p_0,\dots,p_s$ of a spine $P$
	that specify their movement speed.
The point $p_i$ for $i \in \{1,\dots,s\}$ is moved by $\vel(p_i)\varepsilon$.
Thus setting $\vel(p_1)=1$ makes the $\delta$-critical $\{p_0,p_1\}$ become $(\delta+\varepsilon)$-critical, as desired.
Setting $\vel(p_i)=i$ for $i\geq 1$, however,
	can make consecutive points $\{p_{i-1},p_i\}$ uncritical.
\Autoref{figure:accordion} provides an example.
To see this,
	fix some shortest $p_{i-1},p_i$-path $P_i$ and some shortest $p_i,p_{i+1}$-path $P_{i+1}$.
The paths $P_i$ and $P_{i+1}$ can have a trivial intersection of only $\{p_i\}$ or their intersection may contain more than one point.
We denote this bit of information as $\change_P(p_i) \in \{-1,1\}$.
We set $\change_P(p_i)=1$ if and only if $P_i$ and $P_{i-1}$ have a trivial intersection.
(The definition of $\change_P$ is independent on the exact considered shortest paths
	and we will define it properly in the next subsection.)

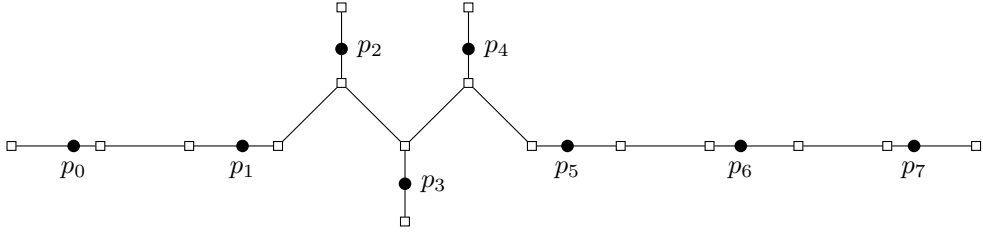
\begin{figure}[t]
\begin{center}
\begin{tikzpicture}
[scale=1.67,auto=left, node/.style={rectangle,fill=white, draw, scale=0.5},
	point/.style={circle,fill=black, scale=0.5}]

	\node [node] (llll) at (-1.1,0) {};	
	\node [node] (lll) at (-0.4,0) {};	
	\node [node] (ll) at (0.3,0) {};
	
	\node [node] (l) at (1,0) {};
	\node [node] (t) at (1.5,{sin 30}) {};
	\node [node] (tt) at (1.5,1.1) {};
	\node [node] (a) at (2,0) {};
	\node [node] (aa) at (2,-0.6) {};
	\node [node] (b) at (2.5,{sin 30}) {};
	\node [node] (bb) at (2.5,1.1) {};
	\node [node] (c) at (3,0) {};
	\node [node] (d) at (3.7,0) {};
	\node [node] (e) at (4.4,0) {};
	\node [node] (f) at (5.1,0) {};
	\node [node] (g) at (5.8,0) {};
	\node [node] (h) at (6.5,0) {};

	\points{p_0}{lll}{llll}{0.3}
	\points{p_1}{l}{ll}{0.4}	
	\pointe{p_2}{t}{tt}{0.45}
	\pointe{p_3}{a}{aa}{0.5}
	\pointe{p_4}{b}{bb}{0.45}
	\points{p_5}{c}{d}{0.4}
	\points{p_6}{e}{f}{0.35}
	\points{p_7}{g}{h}{0.3}
	
	\foreach \from/\to in {llll/lll,lll/ll,ll/l,l/t,t/a,a/b,b/c,c/d,d/e,e/f,f/g,g/h,
		t/tt,a/aa,b/bb}
	\draw (\from) -- (\to);
\end{tikzpicture}

\end{center}
\caption{
A spine $P = (p_0,\dots,p_7)$.
The shortest path between $p_0,p_1$ and the shortest path between $p_1,p_2$ have the trivial intersection of $\{p_1\}$,
	hence $\change_P(p_1) = 1$.
In turn, the shortest path between $p_1,p_2$ and the shortest path between $p_2,p_3$ have a non-trivial intersection,
	hence $\change_P(p_1) = -1$.
Also $\change_P(p_3) = \change_P(p_4) = -1$ while the other values are positive.
Consequently $(\sgn_P(p_1),\dots,\sgn_P(p_7)) = (1,1,-1,1,-1,-1,-1)$.
Thus $(\vel_P(p_1),\dots,\vel_P(p_7)) = (1,2,1,2,1,0,-1)$.
In particular, $\sgn_P(p_5)$ is negative such that it is moved towards $p_4$.
In turn, $\sgn_P(p_7)$ and $\vel_P(p_7)$ are negative such that $p_7$ has a net movement away from $p_0$.
Under this movement all $\{p_0,p_1\}, \dots, \{p_6,p_7\}$ remain critical.
%
}
\label{figure:accordion}
\end{figure}

The easy case is when $P_i$ and $P_{i+1}$ have a trivial intersection, i.e., $\change_P(p_i)=1$.
In this case we increase the velocity of the next point $p_{i+1}$.
The first time we encounter the other case, that $\change_P(p_i)=-1$,
	we decrease the velocity of the next point $p_{i+1}$.
Further we move $p_{i+1}$ \stress{towards} and not away of $p_i$.
Hence we also specify a $\sgn$ of the velocity that records
	whether a point $p_{i+1}$ is pushed towards or away from its predecessor $p_i$.
All these changes are relative to whether the movement of predecessor $p_i$ is away from $p_{i-1}$, i.e., whether $\sgn(p_i)$ is positive.
For example, the second time we encounter a point $p_j$ with $\change_P(p_j)=-1$,
	point $p_{j+1}$ is again moved away from its predecessor.

This leads to the following definition of $\vel_P$ and $\sgn_P$ for a spine $P$.
We define its half-integral \define{velocities} $\vel_P: \{p_0,\dots,p_s\} \allowbreak \to \HHH$
	depending on \define{signs} $\sgn_P: \{p_1,\dots,p_s\} \to \{-1,1\}$,
	which in turn depend on $\change_P$. 
We may drop the subscript $P$, if it is clear from the context.
Let $\vel(p_0)=0$.
Let $\vel(p_1) = \frac{1}{2}$, if $p_0 \in \pivots{S}{\delta}$, and let $\vel(p_1)=1$, if $p_0 \in S$.
For $i \geq 1$, let
$$ \vel(p_{i+1}) \; \coloneqq \; \vel(p_i)+ \sgn(p_{i+1}) . $$
Thus $\sgn \in \{-1,1\}$ indicates whether the velocity increases or decreases.
The current $\sgn$ is unchanged unless $\change$ is negative.
Let $\sgn(p_1) = 1$.
For $2 \leq i \leq s$, let
	$$
	\sgn(p_{i}) \; \coloneqq \; \change(p_{i-1}) \sgn(p_{i-1}) \; = \; \prod_{0< j< i} \change(p_j) . $$
The movement step of a point $p_i$ in a spine $P = (p_0,\dots,p_i)$ is now as follows.
We push the point $p_i$ by the (possibly negative) distance $\sgn_P(p_i) \vel_P(p_i) \varepsilon$ away from its predecessor $p_{i-1}$.
In other words, the point $p_i = p(u,v,\lambda)$ is replaced by the point $p(u,v,\lambda + \sgn_P(p_i) \vel_P(p_i)\varepsilon)$ assuming that vertex $u$ compared to $v$ is in some sense closer to the predecessor point $p_{i-1}$.
We make this notion formal in the next subsection.

\subsection{Directions}

We formalize the notion the direction of a point $p$ towards another point $q$.
The \define{direction} $\ori{q}{p} \in \{u,v\}$ for distinct points $p = p(u,v,\lambda) \in P(G)$ and $q \in P(G)$
	is defined as follows:
\begin{itemize}
\item
For points $p=p(u,v,\lambda_p)$ and $q=p(u,v,\lambda_q)$ on a common edge $\{u,v\} \in E(G)$ with $\lambda_p<\lambda_q$,
	let $\ori{q}{p}=v$.
Let $\nori{q}{p} = u$.
\item
For points $p=p(u_p,v_p,\lambda_p)$ and $q=p(u_q,v_q,\lambda_q)$
	on distinct edges $\{u_p,v_p\} \neq \{u_q,v_q\}$,
	let $\ori{q}{p}$ be the unique vertex of $\{u_p,v_p\}$ that is contained in \emph{every} shortest path between $p$ and $q$,
		if such a vertex exists.
If $\ori{q}{p}$ is defined, let $\nori{q}{p}$ be the unique vertex in $\{u_p,v_p\} \setminus \{ \ori{q}{p} \}$.
\end{itemize}

\newcommand{\textLemmaRoundingOrientation}{
For distinct points $p, q \in V(G_S)$, $\ori{q}{p}$ is well-defined, unless $p$ is half-integral.
}

\begin{lemma}\hyperref[lemma:appendix:rounding:orientation]{\refstar}
\label{lemma:rounding:orientation}
\textLemmaRoundingOrientation
\end{lemma}

Hence we can properly define $\change(p_i)$ of a point $p_i$ of a spine $(p_0,\dots,p_s)$ with $1 \leq i \leq s-1$,
	since $p_i$ with $i\geq 1$ is non-half-integral.
Let $$
		\change(p_i) \coloneqq
		\begin{cases}
			\phantom{-}1, & \ori{p_{i-1}}{p_i} \neq \ori{p_{i+1}}{p_i}, \\ 
			-1, & \text{else}.
		\end{cases}
	$$

Further, for a non-half-integral point $p= p(u,v,\lambda)$ we have $\{u,v\} = \{\ori{q}{p},\nori{q}{p}\}$.
By symmetry assume that $\nori{q}{p} = u$.
We can equivalently specify point $p$ as $p(\nori{q}{p},\ori{q}{p}, \lambda)$.
Conveniently, we may write $p=p(\cdot,\nori{q}{p}, \lambda)$
	since the missing entry is clear from the context.
Doing so, the edge position $\lambda$ measures a part of the length of any shortest $p,q$-path,
	specifically the part using the edge of $p$ (assuming $q$ is on another edge).

Now we can also properly define one pushing step for a point $p_i$ of a spine $P = (p_0,\dots,p_s)$ and for $\varepsilon > 0$.
Let $\lambda_i$ be such that $p_{i} = p(\cdot, \nori{p_{i-1}}{p_{i}}, \lambda_i)$.
Then the new point is
$$
	\label{eq:moved:point} 
	(p_i)_{P,\varepsilon} \coloneqq p \big(\cdot, \; \nori{p_{i-1}}{p_{i}}, \; \lambda_i + \sgn_P(p_i) \vel_P(p_i) \varepsilon \big) .
$$

\newcommand{\lemmaTextStayCritical}{
For a spine $P = (p_0,\dots,p_s)$ and $i \in \{0,\dots,s-1\}$, points $(p_i)_{P,\varepsilon},(p_{i+1})_{P,\varepsilon}$ are $(\delta+\varepsilon)$-critical
	for the maximal $\varepsilon \leq \delta^\star-\delta$ that is limited by the Events 1,2,3.
}

\begin{lemma}\hyperref[lemma:appendix:stay:critical]{\refstar}
\label{lemma:stay:critical}
\lemmaTextStayCritical
\end{lemma}

\subsection{Root Points}

We formally define the set of root points $R$.
Let $R_0$ be the set of half-integral points in $G_S$.
There may be some components of the auxiliary graph $G_S$ without a point in $R_0$.
Let $R$ result from $R_0$ by adding exactly one point from every component that has no point in $R_0$.

We consider only spines $P=(p_0,\dots,p_i)$ where $p_0 \in R$.
Clearly every point $G_S$ is part of at least one spine and hence has some movement prescribed.
We also claim that the prescribed movement is uniquely defined.
In other words, there are no spines $P=(p_0,\dots,p_i)$ and $Q=(q_0,\dots,\allowbreak q_j)$	with $p_0,q_0 \in R$ that terminate at the same point $p_i=q_j$
	and contradict in their prescribed movement for $p_i=q_j$.

\newcommand{\textLemmaEqualPotential}{
Let $\delta < \delta^\star$.
Consider spines $P=(p_0,\dots,p_i)$ and $Q=(q_0,\dots,\allowbreak q_j)$ 
	with $p_0,q_0 \in R$ and $p_i=q_j$. 
Then (1) $\vel_P(p_i) = \vel_Q(q_j)$; and (2) $\nori{p_{i-1}}{p_i} = \nori{q_{j-1}}{q_j}$ if and only if $\sgn_P(p_i) = \sgn_Q(q_j)$.
}

\begin{lemma}\hyperref[lemma:appendix:equal:potential]{\refstar}
\label{lemma:equal:potential}
{\textLemmaEqualPotential}
\end{lemma}
\newcommand{\me}{{\varepsilon^\star}}

Therefore we can uniquely define the moved version of a point $p$ as $p_\varepsilon \coloneqq (p)_{P,\varepsilon}$
	where we may choose $P$ to be an arbitrary spine starting in $R$ and containing $p$.
This defines the set of pushed points $S_\varepsilon \coloneqq \{\, p_\varepsilon \mid p \in S \,\}$.

\medskip

For the proof of \autoref{lemma:equal:potential},
	we use that $\delta^\star = \frac{a^\star}{b^\star} \geq \delta$ is minimal with $a^\star \leq 2L+L$.
Since spines $P,Q$ meet at the point $p_i=q_j$ their roots $p_0,q_0$ must be from the same component in the auxiliary graph $G_S$;
	in other words $p_0,q_0$ are both half-integral or are the same point.
Our proof by contradiction considers these two options and whether $P,Q$ reach $p_i$ and~$q_j$ from the same vertex relatively to if they agree on the $\sgn$ of $p_i=q_j$.
An example is that $p_0=q_0$ and $P,Q$ reach $p_i = q_j$ from the same vertex (formally with $\nori{p_{i-1}}{p_i} \neq \nori{q_{j-1}}{q_j}$) while they agree on the $\sgn$ (formally $\sgn_P(p_i) = \sgn_Q(q_j)$).
Then we can glue $P$ and $Q$ together forming a walk starting in $p_0=q_0$ and returning to $p_0=q_0$.
This walk then has half-integral length at most $2L+2$ but is made up hops of length $\delta$.
That implies the contradiction that already $\delta = \delta^\star$.

\subsection{Summery}

With the previous observations we can assemble the algorithm for \autoref{lemma:overview}.
Our potential counts how many elements can still trigger Events 1,2,3.
That is
$$
\Phi(S) \;\coloneqq\; \Phi_1(S) + \Phi_2(S) + \Phi_3(S)  \;\leq\; 2|S|^2 + |V(G)|^2  .
$$

We define $\me \geq 0$ as the maximal $\varepsilon \leq \delta^\star - \delta$
	limited by the Events 1,2,3.
We claim that $\me \geq 0$ is defined.
This is due to that the above events depend on continuous functions in $\varepsilon$,
	which are the distance of $p_{\varepsilon}$ to its closest half-integral point, and the distance between points $p_\varepsilon$ and $q_{\varepsilon}$ for $p,q \in S$.

To show termination, we prove that no such element can trigger its according event more than once.
\autoref{lemma:stay:critical} already implies that a $\delta$-critical pair of points $\{p,q\}$ stays $(\delta+\me)$-critical.
It remains to show the following monotonicities:

\newcommand{\lemmaTextRoundingInvariant}{
	Let $S$ be a $\delta$-dispersed set for $\delta < \delta^\star$ and $\me$ defined as above.
	Then:
	\begin{itemize}
		\item $(1)$ $S_\me$ is a $(\delta + \me)$-dispersed set of size $|S|$.
		\item $(2)$ If $\{p,q\} \in \binom{S}{2}$ is $\delta$-critical,
			then $\{p_\me, q_\me\}$ is $(\delta + \me)$-critical.
		\item $(3)$ If $r \in \pivots{S}{\delta}$, then $r \in \pivots{S_\me}{\delta+\me}$.
	\end{itemize}
}
\begin{lemma}
\hyperref[lemma:appendix:rounding:invariant]{\refstar}
	\label{lemma:rounding:invariant}
	\lemmaTextRoundingInvariant
\end{lemma}

Now we have all the tools to show \autoref{lemma:overview}.
Determine $\me$ and the $(\delta+\me)$-dispersed set $S_\me$ as defined before.
The resulting set $S_\me$ is a $(\delta+ \me)$-dispersed set of the same size,
	according to \autoref{lemma:rounding:invariant}.
If $\delta + \me = \delta^\star$
	then $S_\me$ is already the desired $\delta^\star$-dispersed set.
Else one of the Events 1,2,3 occurred.
We observe that the potential strictly decreases, that is $\Phi(S_\me) < \Phi(S)$.
Because of the monotonicities of \autoref{lemma:stay:critical} and \autoref{lemma:rounding:invariant} the partial potentials $\Phi_1$, $\Phi_2$ and $\Phi_3$ do not increase.
If \ref{event:1:n} occurs then $\Phi_1$ strictly decrease.
If \ref{event:2:n} occurs then $\Phi_2$ strictly decrease.
If \ref{event:3:n} occurs then $\Phi_3$ strictly decrease.
All in all at least one part strictly decreases and so does $\Phi$.
This completes the proof of \autoref{lemma:overview} and hence of \autoref{lemma:rounding}.

\section{Algorithmic Implications}
\label{section:algo:imp}

Based on the rounding procedure from \autoref{sec:rounding},
	the translation result from \autoref{sec:translating}
	and connections to distance-$d$ independent set we derive a number of algorithmic results.

\begin{theorem}
\label{lemma:L:tw:algo}
There is an algorithm that, given distance $\delta \geq 0$, a graph \( G \), a tree decomposition
and an upper bound $L \in \N$ on the length of the longest path in $G$,
computes a maximum \( \delta \)-dispersed set \( S \) in time \( (2L+2)^{\tw(G)} n^{\Oh(1)} \).
\end{theorem}

\begin{proof}
According to \autoref{lemma:rounding}, we may consider the rounded up distance,
	that is a rational $\frac{a}{b} \geq \delta$ with $a \leq 2L+2$, instead of $\delta$.
Notice that $\frac{a}{b}$ polynomial time computable.
As long as $\frac{a}{b} \leq \tfrac{3}{4}$, we may repeatedly apply \autoref{lemma:limit:b}
	such that eventually we obtain that $\tfrac{3}{4} b < a \leq 2L+2 $.
Let $G_{2b}$ be a $2b$-subdivision of $G$.
Observe that $\tw(G) = \tw(G_{2b})$
	and the number of vertices increases only by a factor of $\Oh(n^2L)$.
According to \autoref{lemma:subdivision:into:is} $\disp{\tfrac{a}{b}}(G) = \alpha_{2a}(G_{2b})$.
Thus, to answer the original $\delta$-\dispersion-instance
	we may find a maximum distance-$2a$ independent set in $G_{2b}$,
	which is possible in time $(2a)^{\tw(G)} n^{\Oh(1)}$, according to \cite{DBLP:conf/wg/KatsikarelisLP18}.
\end{proof}

This result immediately yields parameterized complexity results for the parameters treedepth and treewidth.
Regarding the treewidth, note that \( n \) is an upper bound on \( L \).
Thus the above algorithm is an \xp-algorithm for the parameter treewidth.
When a treewidth decomposition is given,
\textsc{Dispersion} can be solved in time $2n^{\tw(G)} n^{\Oh(1)}$.
	
\begin{corollary} 
\sloppy
	\textsc{Dispersion} can be solved in time $2n^{\tw(G)} n^{\Oh(1)}$,
	assuming a tree decomposition is given.
\end{corollary}

Similarly we obtain an \fpt algorithm for treedepth $\td(G)$ of the input graph.
The treedepth \( \td(G) \) implies a bound on \( L \), which is \( L \leq 2^{\td(G)} \).
Since also \( \td(G) \geq \tw(G) \), we obtain an \( 2^{\Oh(\td(G)^2)} n^{\Oh(1)} \)-time algorithm,
	assuming a treedepth decomposition is given.

\begin{corollary} 
	\label{lemma:dispersion:treedepth}
	\textsc{Dispersion} can be solved in time \( 2^{\Oh(\td(G)^2)} n^{\Oh(1)} \),
	assuming a treedepth decomposition is given.
\end{corollary}

\section{Parameterized Hardness Results}
\label{sec:hardness}

We complement the positive results by hardness results.
These results borrow ideas from hardness-reductions for the similar problem \problemm{Distance Independent Set} (\problemm{DIS}), see \autoref{section:relationship:with:dis}.

A natural generalization of treedepth is the maximum diameter of graph $G$,
	which is the maximum distance between any vertices $u,v \in {V(G)}$
	(since we only consider connected graphs $G$).
We show \np-hardness for graphs of any diameter $\geq 3$ even for chordal graphs
	by a reduction from \problemm{Independent Set},
	similarly as \np-hardness for \problemm{DIS} is shown by Eto et al.~\cite{DBLP:journals/jco/EtoGM14}.
Our reduction also shows \wone-hardness with respect to the solution size $k$.

\newcommand{\lemmaTextHardnessChordal}{
For every $\delta > 3$,
$\delta$-\dispersion is \np-complete and \wone-hard with parameter solution size, even for connected chordal graphs of diameter $\leq \lceil\delta\rceil$.
}

\begin{lemma}\hyperref[lemma:appendix:hardness:chordal]{\refstar}
\label{lemma:hardness:chordal}
\lemmaTextHardnessChordal
\end{lemma}

Another direct generalization of treedepth is pathwidth of the input graph $G$.
We show \wone-hardness even for the combined parameters pathwidth and solution size $\pw(G)+k$.
With the same reduction also \wone-hardness for the combined parameters `feedback vertex set size' $\fvs(G)$ and solution size $k$ follows.
We can essentially use the same reduction as used by 
Katsikarelis et al.\ to show \wone-hardness of \problemm{DIS} when parameterized by $\fvs(G)+k$ by reducing from \problemm{Multi-Colored-Independent-Set}~\cite{DBLP:conf/wg/KatsikarelisLP18}.

\newcommand{\lemmaTextPwFvsETH}{
\dispersion is \wone-hard parameterized by $\pw(G)+k$.
Further, there is no $n^{o(\sqrt{\pw(G)}+ \sqrt{k})}$-time algorithm unless ETH fails.
\dispersion is \wone-hard parameterized by $\fvs(G)+k$.
Further, there is no $n^{o(\fvs(G)+ \sqrt{k})}$-time algorithm unless ETH fails.
}

\begin{theorem}\hyperref[lemma:appendix:pw:fvs:eth]{\refstar}
\label{lemma:pw:fvs:eth}
\lemmaTextPwFvsETH
\end{theorem}

Since $\fvs(G)$ is a linear upper bound for the treewidth of $G$, we also obtain:
\dispersion is \wone-hard parameterized by $\tw(G)+k$.
Further, there is no $n^{o(\tw(G)+ \sqrt{k})}$-time algorithm unless ETH fails.
Similarly as in~\cite{DBLP:conf/wg/KatsikarelisLP18} we obtain a lower bound for treedepth.

\newcommand{\lemmaTextTDlower}{Assuming ETH, there is no $2^{o(\td(G)^2)}$-time algorithm for \dispersion.}

\begin{theorem}\hyperref[lemma:appendix:td:lower:bound]{\refstar}
\label{lemma:td:lower:bound}
\lemmaTextTDlower
\end{theorem}

\section{NP-hardness for Irrational Distance}
\label{sec:irrational}

We show \np-hardness of $\delta$-\dispersion for every irrational distance $\delta >0$.
Thus together with earlier results \cite{grigoriev2021dispersing}
	the complexity for every real $\delta > 0$ is resolved:
For rational distance $\delta = \frac{a}{b}$ where $a \in \{1,2\}$ the problem is polynomial time solvable,
	while it is \np-complete for every other distance $\delta > 0$.

\begin{theorem}
\label{lemma:np:hardness:irrational}
	For every irrational $\delta>0$, $\delta$-\problemm{Dispersion} is \np-complete.
\end{theorem}

The key step is a reduction from \problemm{Independent Set} which
shows \np-hardness not only for a single distance $\delta$ but for the whole interval $\delta \in (2,3]$.

\smallskip
\noindent\textbf{Construction:}
Given a graph $G$ and integer $k \in \N$,
	we construct an input for $\delta$-\problemm{Dispersion} consisting of a graph $G'$ and integer $k' = k$ as follows.
For every vertex $u \in V(G)$ introduce vertices $u_1,u_2$ and edge $\{u_1,u_2\}$.
For every edge $\{u,v\} \in E(G)$ introduce edges $\{u_i,v_j\}$ for every $i,j \in \{1,2\}$.

\begin{lemma}
	\label{lemma:np:hardness:2:3}
	For every $\delta \in (2,3]$,
	$\delta$-\dispersion is \np-hard and \wone-hard when parameterized by solution size.
\end{lemma}

\begin{proof}

Clearly, this construction is polynomial time computable.
Further, the reduction is parameter preserving
	such that \wone-hardness of \problemm{Independent Set} implies \wone-hardness of \dispersion, assuming correctness of the reduction.

Hence, it remains to show the correctness,
 	that $G$ has an independent set of size $k$ if and only if
$G'$ has a $\delta$-dispersed set of size $k$.

\forward
Let $I \subseteq V(G)$ be an independent set of graph $G$.
We define $S \coloneqq \{ p(u_1,u_2,\tfrac{1}{2}) \mid u \in I \} \subseteq P(G)$,
	which has size $|S| = |I|$.
We claim that $S$ is $\delta$-dispersed in $G'$ for $\delta \in (2,3]$.
Since any vertices $u,v \in V(G)$ have distance at least 2 in $G$,
	their corresponding points $p(u_1,u_2,\tfrac{1}{2})$ and $p(v_1,v_2,\tfrac{1}{2})$ have distance at least 3 in $P(G)$.
Thus they are $\delta$-dispersed for $\delta \in (2,3]$.

\backward
Let $S \subseteq P(G)$ be a $\delta$-dispersed for some $\delta \in (2,3]$.
We define the \define{ball} $B_u$ for $u \in V(G)$ as the points in $P(G)$ with distance at most $\tfrac{1}{2}$ to $u_1$ or $u_2$,
	which is $B_u \coloneqq \{ p(u_i,v,\lambda) \mid i \in \{1,2\}, \{u_i,v\} \in E(G'), \lambda \in [0,\tfrac{1}{2}] \}$.
Every ball $B_u$ for $u\in V(G)$ contains at most one point from $S$ since points $p,q \in B_u$ can be at most $2<\delta$ apart.
Every union $B_u \cup B_v$ for adjacent $\{u,v\} \in E(G)$ contains at most one point from $S$ since points $p,q \in B_u \cup B_v$ can also be at most $2<\delta$ apart.

Now we define an independent set $I \subseteq V(G)$.
Add vertex $u \in V(G)$ for every point $p \in S \cap B_u$ except when $p \in B_u \cap B_v$ for some $v\in P(G)$.
If $p \in S \cap B_u \cap B_v$, add either the point $u$ or $v$ to $I$.
Then $|I| = |S|$ since the union of $B_u$ for $u \in V(G)$ is the whole point space $P(G)$.
Further, no adjacent vertices $u,v$ are in $I$ since $B_u \cup B_v$ contain at most one point from $S$.
Thus $I \subseteq V(G)$ is an independent set of size $|S|$.
\end{proof}

Because $\delta \leq 3$ we may apply \autoref{lemma:limit:b:auto} to obtain \np-hardness for $\delta$ in the interval $(\frac{2}{2x+1},\frac{3}{3x+1}]$ for every integer $x \geq 0$.
Applying \autoref{lemma:subdivide} yields \np-hardness for $\delta$ in the interval $(\frac{2c}{2x+1},\frac{3c}{3x+1}]$ for every integer $c \geq 1$. 

Now, consider any irrational distance $\delta > 0$.
Consider $F \coloneqq \{ c {\delta}^{-1} - \lfloor c \delta^{-1} \rfloor \mid c \geq 1\}$,
	the set of fractional parts of multiples of $\delta^{-1}$.
Since $\delta^{-1}$ is irrational, $F$ is a dense subset of the interval $[0,1]$.
Let integer $c \geq 1$ be such that $ \frac{1}{3} \leq c \delta^{-1} - \lfloor c \delta^{-1} \rfloor < \frac{1}{2}$.
Thus there is a non-negative $x$ such that $x+\frac{1}{3} \leq c {\delta}^{-1} < x + \frac{1}{2}$.
This implies that $\frac{2c}{2x+1}< \delta \leq \frac{3c}{3x+1}$
	and hence \np-hardness for $\delta$-dispersion.
This finishes the proof of \autoref{lemma:np:hardness:irrational}.

\section{Parameter Solution Size}
\label{sec:naturalparameter}

$\delta$-\dispersion parameterized by the solution size $k$ is \wone-hard when $\delta > 2$:
When $\delta \in (2,3]$ \autoref{lemma:np:hardness:2:3} shows \wone-hardness,
	while for $\delta > 3$ \autoref{lemma:hardness:chordal} implies \wone-hardness even when the input graph is chordal.
It remains to consider $\delta \leq 2$.
Observe that for \( \delta \leq 2 \), every graph $G$ satisfies \( \disp{\delta}(G) \geq {\nu(G)} \)~\cite{grigoriev2021dispersing},
	where $\nu(G)$ is the maximum matching size of $G$.
Thus, a win-win situation occurs.
Determine $\nu(G)$ in polynomial time.
If $k \leq {\nu(G)}$, we may immediately answer `yes'.
Otherwise $k > {\nu(G)} \geq \tfrac{\vc(G)}{2}$,
	where $\vc(G)$ is the minimum size of a vertex cover in $G$.
The size of a vertex cover upper bounds the treedepth.
A treedepth decomposition of size $\td(G)$ is computable in \fpt-time~\cite{DBLP:conf/icalp/ReidlRVS14}.
Thus we may apply the \fpt algorithm for parameter treedepth from
	\autoref{lemma:L:tw:algo}. 

\begin{theorem}
	$\delta$-\dispersion parameterized by solution size \( k \) is \fpt if \( \delta \leq 2 \);
	and \( W[1] \)-hard if $\delta >2$.
\end{theorem}

\bibliography{literature}

\newpage
\appendix

\section{Translating $\delta$-Dispersion}

\subsection{A Distance Measure}
\newcommand{\dr}{\mathsf{dr}}

We provide a distance measure for an edge $\{u,v\} \in E(G)$ regarding a point set $S \subseteq P(G)$.
This notion will lead us to an alternative definition of auto-dispersion and covering
	that avoids talking about walks in $P(G)$.

For an edge $\{u,v\} \in E(G)$, let $\dr_S(u,v)$ be the minimum length,
	over all points $p \in S$, of a $u,p$-walk in $P(G)$ starting from $u$ in the direction of $v$.

This notion equals a recursive definition that does not involve walks in~$P(G)$.
A point~${p \in S}$ that yields the value $\dr_S(u,v)$,
	is either already on the current edge $\{u,v\}$
	or is found by recursing on another incident edge of $v$.

\begin{lemma}
\label{lemma:dr:recurisive:definition}
For every edge $\{u,v\} \in E(G)$,
$$
	\dr_S(u,v) = \begin{cases}
		\min\{\lambda \mid p(u,v,\lambda) \in S\}, &
		\text{if } S \cap P(G[\{u,v\}]) \neq \emptyset, \\
		1 + \min\{ \dr_S(v,w) \mid w \in N(v) \setminus \{u\} \}, & \text{else}.
	\end{cases}	
$$
\end{lemma}
\begin{proof}
We prove the statement by induction on $\ell' = \lceil \dr_S(u,v) \rceil$.

If $\ell' = 1$, there is a point $p(u,v,\lambda) \in S$
	with edge position $\lambda = \dr_S(u,v)$.
Thus, $\{ p(u,v,\lambda') \mid \lambda' \in [0,\lambda]\} \subseteq P(G)$ is a $u,p$-walk starting from $u$ in the direction of $v$ and of minimum length $\dr_S(u,v)$.

Now assume by induction that the statement holds for all $u',v' \in V(G)$ with ${ \lceil \dr_S(u',v') \rceil < \ell' }$.
Consider $\lceil \dr_S(u,v) \rceil = \ell'$.
Then $\ell' \geq 2$ such that $S \cap P(G[\{u,v\}]) = \emptyset$. 
Let distance $\dr_S(u,v)$ be realized by a $u,p$-walk $P \subseteq P(G)$ starting from $u$ in the direction of $v$.
Since~$P$ has length~$>1$, after traversing edge $\{u,v\}$ it continues along an incident edge $\{v,w\} \in E(G)$ for some neighbor $w \in N(v)\setminus \{u\}$.
Then $P$ consists of $P(G[\{u,v\}])$ and a remaining $v,p$-walk~$P'$ starting from~$v$ in the direction of~$w$.
Thus $\dr_S(v,w) = 1 + \dr_S(v,w)$.

Vice versa, for any neighbor $w' \in N_G(v) \setminus \{u\}$
	there is a $v,p'$-path $Q'$ starting from $u$ in the direction of $w'$ for some point $p' \in S$.
Thus $P(G[\{u,v\}])$ combined with $Q'$ is a $u,p'$-walk $Q \subseteq P(G)$.
By the minimality of $P$, walk~$Q'$ has length at least that of $P'$.
Thus $w$ minimizes $\min_{w' \in N(v) \setminus \{u\}} \dr_S(v,w')$,
	which implies the statement.
\end{proof}

For convenience, let $\ell(p,u) > 0$ be minimum length of a non-trivial $p,q$-walk containing $u$ for any point $q \in S$ (possibly $p=q$).

\begin{lemma}
\label{lemma:auto:conditions:technical}
Let $p = p(u,v,\lambda) \in S$ with $\dr_S(u,v) = \lambda \in (0,1]$.
Then, there exists a neighbor $w \in N_G(u)$ such that
$\dr_S(u,v) + \dr_S(u,w) = \ell(p,u)$.
\end{lemma}
\begin{proof}
Consider a point $p = p(u,v,\lambda) \in S$ with $\dr_S(u,v) = \lambda \in (0,1)$.

($\geq$)
Let $w \in N_G(u) \setminus \{v\}$ be arbitrary.
Let $P'$ realize $\dr_S(u,w)$ for some neighbor,
	which means that $P'$ is $u,q$-walk starting from $u$ in the direction of~$w$ be of length $\dr_S(u,w)$, for a point $q \in S$.
Combine $P'$ and $\{ p(u,v,\lambda) \mid \lambda' \in [0,\lambda) \}$
	to a non-trivial $p,q$-walk $P$ containing $u$ of length $\dr_S(u,v) + \dr_S(u,w)$.
Thus $\dr_S(u,v) + \dr_S(u,w) \geq \ell(p,u)$ for \stress{every} $w \in N_G(u) \setminus \{v\}$.

($\leq$)
Let $P$ realize $\ell(p,u)$,
	which means that $P$ is a non-trivial $p,q$-walk containing $u$ of length $\ell(p,u)$ for some point $q \in S$.
Let neighbor $w \in N_G(u) \setminus~\{v\}$ be
	such that $\{u,w\}$ is the first edge after $\{v,u\}$ visited by~$P$.
Let $P'$ be the subwalk $P$ starting from $u$ and ending in~$q$
	of length $\ell(p,u) - \dr_S(u,v)$.
Since $P'$ is a $u,q$-walk starting from $u$ in the direction of~$w$,
	it has length $\geq \dr_S(u,w)$.
Thus $\dr_S(u,v) + \dr_S(u,w) \leq \ell(p,u)$.
\end{proof}

\subsection{Every Edge Contains at Least One Point}

As a preliminary step, we need the following result.
\begin{lemma}
\label{lemma:delta:1:point}
Let $G$ be a graph and $\delta < 1$.
There exists a maximum $\delta$-dispersed set $S \subseteq P(G)$
	such that every edge $\{u,v\} \in E(G)$ contains a point $p(u,v,\lambda) \in S$ with $\lambda \in (0,1)$.
\end{lemma}
\begin{proof}
Consider a maximum $\delta$-dispersed set $S$.
Then $|S| \geq |E(G)|$,
	since placing a point on every mid-point of each edge is a $\delta$-dispersed set.
For $i\geq0$, let $E_i \subseteq E(G)$ be the set of edges $\{u,v\}$ that contain exactly $i$ points from $S$,
	which is $|S \cap \{ p(u,v,\lambda) \mid 0 \leq \lambda \leq 1 \}| = i$.
Let $E_{\geq i}$ be the union of $E_j$ for $j \geq i$.
Let an edge $\{u,v\} \in E(G)$ be \define{happy} if there is a point $p(u,v,\lambda) \in S$ with $\lambda \in (0,1)$.

Consider an unhappy edge $\{u,v\}$.
We will modify $S$ such that $\{u,v\}$ becomes happy, all happy edges remain happy and $S$ remains equal-sized and $\delta$-dispersed:
Either there will be enough room such that $p(u,v,\tfrac{1}{2})$ may be added,
	or we will find a sequence of point movements that also adds a point on edge $\{u,v\}$.

If already $S \cup p(u,v,\tfrac{1}{2})$ is $\delta$-dispersed,
	it contradicts that $S$ is maximum.
Hence $\delta > \tfrac{1}{2}$, and
	there is at least one point $p_1 \in S$ in distance $< \delta - \tfrac{1}{2} \leq \tfrac{\delta}{2}$ to $\{u,v\}$.
Let point $p_1 = p(w_1,w_2,\lambda_1) \in S$ with $\lambda_1 < \tfrac{\delta}{2}$ have minimal distance to $p(u,v,\tfrac{1}{2})$.
(We even know that $\lambda_1 < \delta - \tfrac{1}{2}$, but this is not essential for the proof).
By symmetry, we may assume that $w_1=v$.
Further, let $w_0 \coloneqq u$.

Based on vertices $w_0,w_1,w_2$ and point $p_1$ we define an \define{exchange sequence} $w_0,w_1,\dots,w_{s+1}$ with points $p(w_i,w_{i+1},\lambda_i)\in S$ for $1 \leq i \leq s$ as follows. 
Assume, we have $w_0,w_1,\dots,w_i$ defined for an $i\geq 2$.
If $\{w_{i-1},w_i\} \in E_{\geq 2}$, stop and set $s \coloneqq i-1$.
Else $\{w_{i-1},w_i\} \in E_1$.
Let $E_{w_i}$ be the subset of edges $\{w_i,w_{i+1}\} \in E_{\geq 1} \setminus \{ \{w_{i-1}, w_i\} \}$ that are incident to $w_i$.
If $E_{w_i}$ is empty, stop and set $s \coloneqq i-1$.
Otherwise, let $p_i = p(w_i,w_{i+1},\lambda_i) \in S$ be a point on an edge $\{w_i,w_{i+1}\} \in E_{w_i}$
	where $p_i$ has minimal distance to vertex $w_i$.
If $\lambda_i > \lambda_{i-1}$ or $w_{i+1} \in \{w_1,\dots,w_{i}\}$, stop and set $s \coloneqq i$.
Else, continue with increased $i$.

\medskip

We denote an exchange sequence $P = w_0,w_1,\dots,w_{s+1}$ as \define{non-increasing}
	if $\lambda_s \leq \dots \leq \lambda_1 < \tfrac{\delta}{2}$.
We claim that we can modify $S$
	such that every exchange sequence is non-increasing.
We achieve this goal by moving a point closer to the midpoint of its edge.

Consider an exchange sequence $P = w_0,w_1,\dots,w_{s+1}$ that is not non-increasing.
By construction, we have that $\lambda_{s} > \lambda_{s-1} \leq \dots \leq \lambda_1 < \tfrac{\delta}{2}$ for $s\geq 2$.
We replace the second to last point $p_{s-1} = p(w_{s-1},w_s, \lambda_{s-1} )$ by a new point $p_{s-1}' = p(w_{s-1},w_s, \min\{ \tfrac{\delta}{2}, \lambda_s \} )$.
Observe that $S$ under this modification remains equal-sized, $\delta$-dispersed,
	and no happy edge becomes unhappy.
After this modification, if $s \geq 3$, the shortened sequence $w_0,w_1,\dots,w_{s}$ is an exchange sequence that terminates with a point $p'_{s-1} = p(w_{s-1}, w_s, \lambda'_{s-1})$.
Possibly $\lambda'_{s-1} > \lambda_{s-2}$, in which case this exchange sequence is also not non-increasing.
We apply this modification iteratively, until every exchange sequence starting with $w_0=u, w_1=v$ is non-increasing.

We claim that this process terminates.
If at one step $p_{s-1}' = p(u,v, \tfrac{\delta}{2} )$,
	then it cannot be the second to last point of an exchange sequence again,
	and hence will not be moved again.
Else $p_{s-1}'$ compared to $p_{s-1}$ is moved by $\lambda_\varepsilon \coloneqq \lambda_{s}-\lambda_{s-1}$ closer to $\tfrac{\delta}{2}$.
Let 
$$
\Lambda = \{ \lambda \in (0,1) \mid p(u,v,\lambda) \in S, \{u,v\} \in E(G) \}
$$
be fixed before starting 
this modification process. 
Throughout this process it holds that $\lambda_\varepsilon \geq ( \prod_{\lambda \in \Lambda} \lambda )^{-1}$,
since $\lambda_\varepsilon > 0$ and $\lambda_\varepsilon = \sum_{\lambda \in \Lambda} m_{\lambda} \lambda$ with $m_{\lambda} \in \mathbb{Z}$ for each $\lambda \in \Lambda$.
Thus after finitely many moves of $p_{s-1}$, it reaches edges position $\tfrac{\delta}{2}$.

If now no exchange sequence starting with $w_0=u, w_1=v$ exists,
apply the process to $w_0=v, w_1=u$. If also in this case, 
no exchange sequence starting with $w_0=v, w_1=u$ exists,
there is no point of $S$ in distance $\delta$ to $p(u,v,\tfrac{1}{2})$.
A contradiction to $S$ being maximum.

\medskip

Now, every exchange sequence $P$ starting with $w_0=u, w_1=v$ is non-increasing,
	in other words, $\lambda_s \leq \dots \leq \lambda_1 < \tfrac{\delta}{2} < \tfrac{1}{2}$,
	and such an exchange sequence exists.
Then, we may \define{apply} such an exchange sequence $w_0,w_1,\dots,w_s$ to $S$ as follows.
For $1 \leq i \leq s$, replace every point $p_i = p(w_i,w_{i+1},\lambda_i) \in S$ with the point $p_i' = p(w_{i-1},w_i,1-\lambda_i)$,
	which is the point `mirrored' at~$w_i$.
It is easy to see that $S$ remains equal-sized and $\delta$-dispersed under this modification.

Now consider an exchange sequence $P = w_0,w_1,\dots,w_{s+1}$.
We distinguish the following cases.
\begin{itemize}
\item
If $P$ terminates because $\{w_{s},w_{s+1}\} \in E_{\geq 2}$,
apply $P$ to $S$ resulting in $S' \subseteq P(G)$.
Then~$S$ under this modification has $\{w_0,w_1\}$ newly happy
	and every happy edge remains happy, beside possibly $\{w_{s},w_{s+1}\}$.
If $\{w_{s},w_{s+1}\}$ became unhappy,
	it contained only the point $p_{s} \in S$ and point $q=p(w_{s},w_{s+1},1) \in S$ on the vertex $w_{s+1}$.
Replace $q$ by point $p(w_{s},w_{s+1},\delta)$.
Then, $S$ remains equal-sized and $\delta$-dispersed
	and edge $\{w_s,w_{s+1}\}$ is happy.
	
\item
If $P$ terminates because $E_{w_{s+1}}$ is empty,
then $w_{s+1}$ is only incident to edge $\{w_s,w_{s+1}\}$ and edges from $E_0$.
Apply $P$ to $S$ resulting in ${S' \subseteq P(G)}$.
Then $S$ under this modification is $\delta$-dispersed, $\{w_0,w_1\}$ is newly happy
	and every happy edge remains happy, except for $\{w_{s},w_{s+1}\}$ which becomes unhappy.
There is no point of $S$ in distance $\delta$ to $p(w_{s},w_{s+1},\delta)$.
A contradiction to the maximality of $S$.
\item
If $P$ terminates because $w_{s+1} = w_{k} \in \{w_1,\dots,w_{s}\}$,
	replace every point $p_i$ by a point $p_i' = p(w_i,w_{i+1}, \tfrac{\delta}{2})$ for $1\leq i \leq s$, resulting in a set $S'$.
Clearly, $|S'| = |S|$ and every edge for $S$ is also happy for $S'$.

We claim that $S'$ is $\delta$-dispersed.
The points $p_i', p_{i+1}'$ for $i \in \{1,\dots,s\}$ have distance $1$.
Further, point $p_s'$ has distance $1$ to $p_k$ and distance $2-\tfrac{\delta}{2}$ to $p_{k-1}$ (if $k\geq 2$).
It follows that $p_1,\dots,p_s$ are $\delta$-dispersed.
Assume, for the sake of contradiction, that there is a point $q \in {S' \setminus \{p_1',\dots,p_w'\}}$ that has a shortest $q,p_i'$-path for some $i \in \{1,\dots,s\}$.
Since point~$p_i'$ compared to $p_i$ is closer to $w_{i+1}$, path $P$ uses vertex $w_{i+1}$.
Thus, $q$ is at position $q=p(w_{i+1},w',\lambda)$ for some neighbor $w' \in N_G(w_{i+1})$ and $\lambda < \tfrac{\delta}{2}$.
Because vertex $w_{i+1}$ has at most one neighbor $w'$ that contains a point in distance $< \tfrac{\delta}{2}$ to $w_{i+1}$, it must be that $w' = w_{i+2}$ and $i< s$.
A contradiction to that $S'$ contains only point $p_{i+1}'$ on edge $\{w_{i+1},w_{i+2}\}$.

After this modification, the closest point in $S'$ to vertex $u$ is $p'$ with distance $\tfrac{\delta}{2}$.
Consider exchange paths starting with $w_0=v$ and $w_1=u$ instead and apply above modifications
Possibly, the closest point in the updated set of points has distance $\tfrac{\delta}{2}$, such that a point $p(u,v,\tfrac{1}{2})$ can be added, which makes also edge $\{u,v\}$ happy.
Otherwise, some construction made edge $\{u,v\}$ happy.
\end{itemize}
By repeatedly applying this procedure for every unhappy edge $\{u,v\}$, 
eventually we obtain a $\delta$-dispersed set of size $|S|$ without unhappy edges.
\end{proof}

\begin{corollary}
	\label{lemma:auto:delta:1:point}
	Let $\delta < 1$.
	Then there exists a $\delta$-auto-dispersed set $S$ such that $|S| = \dispauto{\delta}(G)$ and for every edge $\{u,v\} \in E(G)$ there is a point $p = p(u,v,\lambda) \in S$ with $\lambda \in (0,1)$.
\end{corollary}

\subsection{An Alternative Definition}

We obtain the following equivalence to auto-dispersion.
A set of points $S \subseteq P(G)$ is \define{edge internally $\delta$-dispersed}
	if every edge by itself does not falsify that $S$ is $\delta$-dispersed,
	formally that every distinct points $p,q \in P(G[\{u,v\}])$ for every edges $\{u,v\}\in E(G)$
	are $\delta$-dispersed.

\begin{lemma}
\label{lemma:auto:dispersion:conditions}
Let $G$ be a graph and $\delta > 0$.
A subset $S \subseteq P(G)$ is $\delta$-auto-dispersed if and only if
\begin{enumerate}
\item \labeltext{$($A1$)$}{condition:a1}
	$S$ is edge internally $\delta$-dispersed,
\item \labeltext{$($A2$)$}{condition:a2}
$\dr_{S}(u,v) + \dr_{S}(u,w) \geq \delta$ for every distinct adjacent edges $\{u,v\},\{u,w\} \in E(G)$ where vertex $u$ is not in $S$.
\end{enumerate}
\end{lemma}
\begin{proof}
\forward
Let $S \subseteq P(G)$ be $\delta$-auto-dispersed.
Then clearly $S$ is edge internally $\delta$-dispersed. 
Now assume, for the sake of contradiction, that \ref{condition:a2} is violated,
	which means that
	there are distinct adjacent edges $\{u,v\},\{u,w\} \in E(G)$ where vertex $u$ is not in $S$ and $\dr_{S}(u,v) + \dr_{S}(u,w) = \ell < \delta$.

If $\dr_S(u,v) \leq 1$ there is a point $p \in S$ on edge $p(u,v,\lambda)$ with $\lambda > 0$ since $u$ is not in~$S$.
Then \autoref{lemma:auto:conditions:technical}
	implies that $\ell(u,p) = \dr_{S}(u,v) + \dr_{S}(u,w) < \delta$,
	hence that there is non-trivial $p,q$-walk containing $u$ of length $<\delta$ for some point $q \in S$.
A contradiction to that $S$ is $\delta$-dispersed.
Symmetrically, assuming $\dr_S(u,w) \leq 1$ leads to a contradiction.
	
Thus, it remains to consider that $1 < \dr_S(u,v) \leq \dr_S(u,w)$, up to symmetry.
Then $\dr_S(u,v) = \min_{v' \in N_G(v)} 1+ \dr_S(v,v')$ according to \autoref{lemma:dr:recurisive:definition}.
Thus, $\dr_S(v,v') = \dr_S(u,v)-1$.
Consider the pair of edges $(v,v'), (v,u)$ instead,
	which has at least $\dr_S(v,v')$ smaller than $\dr_S(u,v), \dr_S(u,w)$.
We repeat this argument until $\dr_S(v,v') <1$.
Then again the previous case applies.
Therefore, also condition \ref{condition:a2} is satisfied.

\backward
Let $S \subseteq P(G)$ be not $\delta$-dispersed.
Then there are points $p,q \in S$ (possibly~$p=q$) with a non-trivial $p,q$-walk in $P(G)$.
If $p,q$ are distinct and from the same edge ${ \{u,v\} \in E(G) }$, hence $p,q\in P(G[\{u,v\}])$,
	then $S$ is not edge internally $\delta$-dispersed.

Else, we may assume that $p=p(u,v,\lambda)$ for some edge $\{u,v\} \in E(G)$ for $\lambda \in (0,1]$ and a shortest non-trivial $p,q$-walk contains $u$.
In particular $\lambda = \dr_S(u,v)$ since otherwise we may consider $p',q' \in S$ with smaller distance than $p,q$.
Then $\ell(p,u) \leq \ell < \delta$ by definition.
According to \autoref{lemma:auto:conditions:technical},
	$\dr_S(u,v) + \dr_S(u,w) = \ell < \delta$
	for some neighbor $w \in N_G(u)$,
	which violates condition~\ref{condition:a2}.
\end{proof}

\subsection{Constructive Part}

Now, we are ready to prove the actual translation lemma.

\begin{lemma}[(\autoref{lemma:limit:b:auto} repeated)]
\label{lemma:appendix:limit:b:auto}
\lemmaTextLimitBAuto
\end{lemma}
\begin{proof}
($\leq$)
First, we show that $\dispauto{\delta}(G) \leq \dispauto{\frac{\delta}{\delta+1}}(G) + |E(G)|$.
Consider a $\delta$-dispersed set $S \subseteq P(G)$ of size $|S| = \disp{\delta}(G)$.
Let $\rho = (\delta+1)^{-1}$, the conversion ratio.
We construct a $(\delta\rho)$-dispersed set $S'$ of size $|S'| = |S|+|E(G)|$. 
A key observation is that $$\rho + \delta\rho \;=\; 1 .$$

We construct $S'$ by considering each edge $\{u,v\} \in E(G)$ separately.
Set~$S'$ compared to~$S$
	will contain one more point in $P(G[\{u,v\}])$ for every edge $\{u,v\} \in E(G)$.
Any point~$p \in S$ on a vertex will also be in~$S'$.
Thus $|S'| = |S| + |E(G)|$.
\begin{itemize}
\item
Consider that edge $\{u,v\} \in E(G)$ contains $m \geq 1$ points from $S$,
	which is $|S \cap P(G[\{u,v\}])| = m$.
Let $p = p(u,v,\lambda)$ be the point among $S \cap P(G[\{u,v\}])$ with minimum distance to~$u$, hence with distance $\lambda$.
Analogously, let $q= p(u,v, 1- \mu)$ be the point among $S \cap P(G[\{u,v\}])$ with minimal distance to $v$,
	hence with distance $\mu$,
	possibly $p=q$.
Then $\dr_S(u,v)=\lambda$ and $\dr_S(v,u)=\mu$.
There are $m-1 \geq 0$ further points between $p$ and $q$ on edge $\{u,v\}$,
	such that $p$ and $q$ have distance $1-(\lambda + \mu) \geq m\delta$.
In other words $\lambda + \mu \leq 1 - m\delta$.
	
Add points $p' = p(u,v,\lambda \rho)$ and $q' = p(u,v, 1 - \mu \rho)$ to $S'$,
	which are distinct even in when $p=q$.
Note that if $p$ is positioned on a vertex, also $p'$ is; analogously for $q$ and~$q'$.
Thus, the distance between the new points $p'$ and $q'$ is
	$$1 - (\lambda + \mu)\rho \; \geq \; 1 - (1-m\delta) \rho \; = \; 1 - \rho + m \delta\rho \; = \; (m+1)\delta\rho . $$
Hence, we may add $m$ further points between $p'$ and $q'$ on edge $\{u,v\}$ to set $S'$
	in such a way that the now $m+2$ points in $S' \cap P(G[\{u,v\}])$ have pairwise distance at least~$\delta\rho$.
\item
It remains to consider that edge $\{u,v\} \in E(G)$ contains no point from $S$,
	which is $S \cap P(G[\{u,v\}]) = \emptyset$.
Fix a neighbor $w_u \in N(u)$ with minimum distance $\dr_S(u,w_u)$.
Analogously, fix a neighbor $w_v \in N(v)$ with minimum distance $\dr_S(v,w_v)$.
By symmetry, we may assume the inequality $\dr_S(u,w_u) \leq \dr_S(v,w_v)$.
Then $w_u \neq v$ since then $\dr_S(u,w_u) = 1 + \dr_S(v,w_v)$.
Then, add a new point $p(u,v, \lambda' )$ at edge position $\lambda' = \max\{ \frac{\delta\rho}{2}, \delta\rho - \dr_S(u,w_u)\rho\}$ to~$S'$.
\end{itemize}

\sloppy
We show that $S'$ is $(\delta\rho)$-auto-dispersed
	by proving the conditions of \autoref{lemma:auto:dispersion:conditions}.
Set $S'$ is edge internally $(\delta\rho)$-dispersed
	since for every edge we have added points that have pairwise distance at least $\delta$.
It remains to show condition~\ref{condition:a2},
	which is that $\dr_{S'}(u,v) + \dr_{S'}(u,w) \geq \delta$ for every distinct adjacent edges $\{u,v\},\{u,w\} \in E(G)$ where vertex $u$ is not in $S$.

We partition the directed edges $(u,v)$ for $\{u,v\} \in E(G)$ as follows:
\begin{itemize}
\item 
A directed edge $(u,v)$ is \define{positive} if it satisfies $\dr_{S'}(u,v) \geq \dr_S(u,v)\rho$.
This is the case when $S \cap G[\{u,v\}]\neq \emptyset$, by construction.
Further, we claim that $(u,v)$ is also positive if $S \cap G[\{u,v\}] = \emptyset$
	and point $p(u,v,\lambda')$ with $\lambda' = 1 - (\delta\rho - \dr_S(v,w_v)\rho)$ is added to~$S'$,
	where $w_v \in N_G(v)$ is a neighbor with minimum $\dr_S(v,w_v)$,
	which has $w_v \neq u$.
Then $\dr_S(u,v) = 1 + \dr_S(v,w_v)$.
It follows that
$$
	\dr_{S'}(u,v) \;=\; 1 - \delta\rho + \dr_S(v,w_v)\rho
	\;=\; 1 - \delta\rho + \dr_S(u,v)\rho - \rho
	\;=\; \dr_S(u,v)\rho .
$$
\item
A not positive directed edge $(u,v)$ is \define{neutral} if $S \cap G[\{u,v\}] = \emptyset$ and point $p(u,v,\lambda')$ with $\lambda' = 1-\frac{\delta\rho}{2}$ is added to $S'$.
Then, since it is not positive, $\dr_S(v,w_v)\rho > \frac{\delta\rho}{2}$
	where $w_v \in N_G(u)$ is a neighbor with minimum $\dr_S(v,w_v)$.
Also $w_v \neq u$.
\item
A directed edge $(u,v)$ is \define{negative} if $S \cap G[\{u,v\}] = \emptyset$ and point $p(u,v, \lambda')$
	with $\lambda' = \max\{\tfrac{\delta\rho}{2}, \delta\rho - \dr_S(u,w_u)\rho \}$ is added to $S'$,
	in which case $w_v \neq u$.
\end{itemize}
Now, we observe that distinct adjacent edges $\{u,v\},\{u,w\} \in E(G)$ satisfy condition~\ref{condition:a2},
	that is $\dr_{S'}(u,v) + \dr_{S'}(u,w) \geq 2(\frac{\delta\rho}{2}) \geq \delta\rho$.
We distinguish which type of $(u,v)$ and $(u,w)$ have.
\begin{itemize}
\item
If $(u,v),(u,w)$ are positive, $\dr_{S'}(u,v) + \dr_{S'}(u,w) \geq (\dr_{S}(u,v) + \dr_{S}(u,w))\rho \geq \delta \rho = \delta\rho$.
Here we used that \autoref{lemma:auto:dispersion:conditions} applies to $S$ and $\delta$,
	hence that $\dr_{S}(u,v) + \dr_{S}(u,w) \geq \delta$.
\item
If $(u,v),(u,w)$ are both not positive,
	then $\dr_{S'}(u,v) + \dr_{S'}(u,w) \geq 2\frac{\delta\rho}{2} \geq \delta\rho$,
	which follows from the fact that $\delta\rho < 1$.
\item
Consider that $(u,v)$ is neutral and $(u,w)$ is positive.
By construction, we have $\dr_S(u,w) \geq \dr_S(u,w_u) \geq \dr_S(v,w_v)$.
Then
\begin{align*}
\dr_{S'}(u,w)+\dr_{S'}(u,v)
\;\geq&\; \dr_{S}(u,w)\rho + 1 - \tfrac{\delta\rho}{2} \\
\;\geq&\; \dr_{S}(v,w_v)\rho + 1 - \tfrac{\delta\rho}{2} \\
\;\geq&\; 1 \\
\;>&\; \delta\rho.
\end{align*}
\item
It remains to consider that $(u,v)$ is negative and $(u,w)$ is positive.
Observe that $\dr_{S'}(u,w) + \dr_{S'}(u,v) \geq \dr_{S'}(u,w_u) + \dr_{S'}(u,v)$.
If $(u,w_u)$ is neutral or negative, we obtain that $ \dr_{S'}(u,w_u) + \dr_{S'}(u,v) \geq \delta\rho$, by the previous cases, as desired.
Thus, consider that $(u,w_u)$ is positive.
Then
\begin{align*}
\dr_{S'}(u,w_u) + \dr_{S'}(u,v)
\;\geq\;& \dr_{S'}(u,w_u) + \delta\rho - \dr_{S}(u,w_u)\rho \\
\;\geq\;& \dr_{S'}(u,w_u) + \delta\rho - \dr_{S'}(u,w_u)
\;=\; \delta\rho. 
\end{align*}
\end{itemize}
Therefore, \autoref{lemma:auto:dispersion:conditions} applies to $S'$ and $\delta\rho$,
	which shows that $S'$ is $\delta\rho$-auto-dispersed.

\bigskip

($\geq$)
Now, we show that $\dispauto{\delta}(G) \geq \dispauto{\frac{\delta}{\delta+1}}(G) + |E(G)|$.
Let $\rho^{-1} = (\delta+1)$, the conversion ratio.
Consider a $(\delta\rho)$-dispersed set $S' \subseteq P(G)$ of size $|S'| = \dispauto{\delta}(G)$.
We construct a $\delta$-covering set $S \subseteq P(G)$ of size $|S| = |S'|-|E(G)|$.
A key observation is that $$\rho^{-1}+\delta\rho \;=\; 1 .$$

Similarly to before, we construct $S$ by considering each edge $\{u,v\} \in E(G)$ separately.
We may assume that every edge contains at least one point,
	according to \autoref{lemma:auto:delta:1:point}.
Set~$S$ compared to~$S'$
	will contain one less point $p(u,v,\lambda)$ with $\lambda \in (0,1)$ for every edge $\{u,v\}$.
Any point $p \in S'$ on a vertex will also be in $S$.
Thus $|S| = |S'| - |E(G)|$.
\begin{itemize}
\item
Consider that edge $\{u,v\}$ contains exactly one point from $S'$,
		which is $S' \cap P(G[\{u,v\}]) = \emptyset$.
Then add no point for edge $\{u,v\}$ to $S$.
\item
Consider that edge $\{u,v\}$ contains $m+1 \geq 2$ points from $S'$,
	which is $|S' \cap P(G[\{u,v\}])| = m+1$, for $m\geq 1$.
Let $p' = p(u,v,\lambda)$ be the point among $S' \cap P(G[\{u,v\}])$ with minimal distance to vertex~$u$,
	hence with distance $\lambda$.
Analogously, let $q' = p(u,v,1-\mu)$ be the point among $S' \cap P(G[\{u,v\}])$ with minimal distance to $v$,
	hence with distance $\mu$.
Then $\dr_{S'}(u,v)=\lambda$ and $\dr_{S'}(v,u)=\mu$.
We have that $p' \neq q'$.
	
Add points $p = p(u,v, \lambda \rho^{-1})$ and $q = p(u,v, 1- \mu \rho^{-1})$ to $S$.
If $p'$ is on a vertex, also the new point $p$ is on the same vertex;
	analogously for $q'$ and $q$.

Observe that the distance between the old points $p'$ and $q'$ is $1-(\lambda+\mu) \geq \delta\rho m$
	since points $S' \cap G[\{u,v\}]$ have pairwise distance at least $\delta$.
Thus, the distance between the new points $p$ and $q$ is
$$
	1 - (\lambda + \mu)\rho^{-1}
	\; = \; (1-(\lambda + \mu))\rho^{-1} + (1-\rho^{-1})
	\; \geq \; \delta\rho  m  \rho^{-1} - \delta
	\; = \; \delta (m-1) .$$

If $m=1$, then $1-(\mu+\lambda) \geq \delta\rho$,
	and hence the distance between $p$ and $q$ is $0$.
In other words, $p = q$
		such that $S$ contains one point less on edge $\{u,v\}$.
Else, if $m \geq 2$,
	we may add $m-2 \geq 0$ further points between $p$ and $q$
	such that the $m$ points $S \cap P(G[\{u,v\}])$ have pairwise distance at least $\delta$.
Then, $S$ contains one more point on edged $\{u,v\}$ than~$S'$.
\end{itemize}

We prove that $S$ is $\delta$-auto-dispersed
	by showing the conditions of \autoref{lemma:auto:dispersion:conditions}.
Set $S$ is edge internally $\delta$-dispersed
	since for every edge we have added points that have pairwise distance at least $\delta$.
For condition~\ref{condition:a2}
	it suffices to show that $\dr_{S}(u,v) \geq \dr_{S'}(u,v) \rho^{-1}$ for every pair~$(u,v)$ where $\{u,v\} \in E(G)$.
Then for every distinct adjacent edges $\{u,v\},\{u,w\} \in E(G)$
	we have $\dr_S(u,v) + \dr_S(u,w) = (\dr_S(u,v)+\dr_S(u,w))\rho^{-1} \geq \delta\rho \rho^{-1} = \delta$ since \autoref{lemma:auto:dispersion:conditions} applies to the $(\delta\rho)$-auto-dispersed set $S'$.

Hence, assume, for the sake of contradiction,
	that there are $(u,v)$ with $\{u,v\}\in E(G)$ where $\dr_{S}(u,v) < \dr_{S'}(u,v) \rho^{-1}$.
Among all such pairs $(u,v)$ where $\dr_{S}(u,v) < \dr_{S'}(u,v) \rho^{-1}$,
	let $(u,v)$ be where $\dr_S(u,v)$ is minimum.

Consider that $\dr_S(u,v) \leq 1$.
Then $S$ contains a point $p = p(u,v, \dr_S(u,v))$.
By construction,
	$S'$ contains point $p' = p(u,v,\lambda')$ with $\lambda' =  \dr_{S}(u,v)\rho = \dr_{S'}(u,v)$.
Thus, actually, $\dr_{S}(u,v) \leq \dr_{S'}(u,v) \rho^{-1}$.

Otherwise, we have $\dr_S(u,v) > 1$.
Then $\dr_S(u,v) = 1 + \dr_S(v,w)$ for some neighbor ${ w \in N_{G}(v) \setminus \{u\} }$.
Because $\dr_S(u,v) > \dr_S(v,w)$,
	the assumption $\dr_S(v,w) \geq \dr_{S'}(v,w) \rho^{-1}$ holds.
We will follow that $\dr_S(u,v) \geq \dr_S(u,v) \rho^{-1}$.
Again, we use that \autoref{lemma:auto:dispersion:conditions} applies to~$S'$,
	such that $\dr_{S'}(v,u) \geq \delta\rho - \dr_{S'}(v,w)$.
By construction, set $S' \cap P(G[\{u,v\}])$ contains only one point.
Thus $\dr_{S'}(v,u) = 1 - \dr_{S'}(u,v)$.
It follows that
\begin{align*}
\dr_S(u,v)
	\;=&\; 1 + \dr_S(v,w) \\
	\;\geq&\; (1-\delta\rho + \dr_{S'}(v,w)) \rho^{-1} \\
	\;\geq&\; (1 - \dr_{S'}(v,u))\rho^{-1} \\
	\;=&\; \dr_S(u,v) \rho^{-1}.
\end{align*}
Hence, in both cases, we obtain the contradiction that $\dr_S(u,v) \geq \dr_S(u,v) \rho^{-1}$.
Thus, by \autoref{lemma:auto:dispersion:conditions}, $S$ is $\delta$-auto-dispersed.
\end{proof}

\section{Rounding the Distance}

\subsection{Proof of \autoref{lemma:rounding:orientation}}

We rely on the \define{direction} $\ori{q}{p} \in \{u,v\}$ for distinct points $p,q \in P(G)$
	defined as follows:
\begin{itemize}
\item
For points $p=p(u,v,\lambda_p)$ and $q=p(u,v,\lambda_q)$ on a common edge $\{u,v\}$ with $\lambda_p<\lambda_q$,
	let $\ori{q}{p}=v$.
Let $\nori{q}{p} = u$.
\item
For points $p=p(u_p,v_p,\lambda_p)$ and $q=p(u_q,v_q,\lambda_q)$
	on distinct edges $\{u_p,v_p\} \neq \{u_q,v_q\}$,
	let $\ori{q}{p}$ be the unique vertex of $\{u_p,v_p\}$ that is contained in \emph{every} shortest path between $p$ and $q$,
		if such vertex exists.
If $\ori{q}{p}$ is defined, let $\nori{q}{p}$ be the unique vertex in $\{u_p,v_p\} \setminus \{ \ori{q}{p} \}$.
\end{itemize}

\begin{lemma}[(\autoref{lemma:rounding:orientation} restated)]
\label{lemma:appendix:rounding:orientation}
\textLemmaRoundingOrientation
\end{lemma}
\begin{proof}
Let $p=p(u_p,v_p,\lambda_p)$ and $q=p(u_q,v_q,\lambda_q)$.
If $\{u_p,v_p\} = \{u_q,v_q\}$, then $\ori{q}{p}$ is defined for $p,q$.
It remains to consider that $\{u_p,v_p\}\neq\{u_p,v_q\}$.
Assume, for the sake of contradiction, that there are shortest paths $P_u, P_v \in P(G)$
	which contain $u_p$ respectively $v_p$.
Let \( \ell_u \in \N \) be the length of the path \( P_u \) without edges $\{u_p,v_p\}$ and $\{u_q,v_q\}$.
Analogously let \( \ell_v \in \N \) be the length of the path $P_v$ without edges $\{u_p,v_p\}$ and $\{u_q,v_q\}$.
	
If $P_u$ and $P_v$ both contain $u_q$ or both contain $v_q$, 
	then their length may only differ by $|\lambda_p - (1-\lambda_p)|<1$.
Since $P_u$ and $P_v$ have equal length, $\lambda_p = \frac{1}{2}$ in contradiction that $p$ is not half-integral.
Thus, up to symmetry, $P_u$ contains $u_q$ and $P_v$ contains $v_q$.

First, assume that $q \in S \setminus \pivots{S}{\delta}$,
	such that $p,q$ have distance $\delta$.	
Then the distance between $p$ and $q$ also computes as \( \lambda_p + \ell_u + \lambda_q = (1-\lambda_p) + \ell_v + (1-\lambda_q) = \delta \).
In other words $2+\ell_u+\ell_v = 2 \delta$,
	which implies that $\delta$ is half-integral.
	Contradiction to \autoref{lemma:not:half:integral}.
	
Now, assume that $q \in \pivots{S}{\delta}$.
Then $\lambda_q \in \{0,\tfrac{1}{2},1\}$.
The length of path $P_v$ is $\lambda_p + x$ for some $x \in \HHH$.
Similarly the length of path $P_u$ is $(1-\lambda_p) + y$ for some $y \in \HHH$.
Since $\lambda_p$ is not a half integral, paths $P_u$ and $P_v$ have different length.
Contradiction.
\end{proof}

\newcommand{\lemmaTextDeltaNotHalfIntegral}{
$\delta$ is not half-integral.}

\subsection{Proof of \autoref{lemma:stay:critical}}

We work towards proving \autoref{lemma:stay:critical}.
We observe the following as an underlying principle.

\newcommand{\textLemmaRoundingOne}{
Consider points
	$p = p(u_p,v_p,\lambda_p)$, $q=p(u_q,v_q,\lambda_q)$
	in distance $\delta' > 0$ and a shortest $p,q$-path that does not visit $v_p,v_q$.
Let $x,y \in \mathbb{R}$, $x+y=1$ and $\varepsilon \geq 0$.
For $\varepsilon' \in [0, \varepsilon)$, let no point $p(u_p,v_p,\lambda_p + x \varepsilon')$, $p(u_q,v_q, \lambda_q + y \varepsilon')$  be half-integral nor $\delta' + \varepsilon'$ be half-integral.
Then points
	$p(u_p,v_p,\lambda_p + x \varepsilon)$ and $p(u_q,v_q, \lambda_q + y \varepsilon)$ have distance $\delta'+\varepsilon$.
}

\begin{lemma}
\label{lemma:rounding:one} 
\label{lemma:appendix:rounding:one}
\textLemmaRoundingOne
\end{lemma}
\begin{proof}
	Let $p_{\varepsilon} \coloneqq p(u_p, v_p, \lambda_p + x \varepsilon)$ and 
	$q_{\varepsilon} \coloneqq p(u_q, v_q, \lambda_q + y \varepsilon)$.

The simple case is that points $p,q$ are from a common edge $\{u,v\}$.
Then any shortest path between $p_{\varepsilon},q_{\varepsilon}$ uses edge $\{u,v\}$
	such that points $p_{\varepsilon},q_{\varepsilon}$ have distance $\delta' + x\varepsilon + y\varepsilon = \delta' + \varepsilon$, as desired.

The interesting case is that $p,q$ are from distinct edges.
Then, there is a shortest path $P\subseteq P(G)$ of length $\delta'$ between the points $p$ and $q$
	that uses vertices $u_p$ and $u_q$.
Let $Q \subseteq P$ be the subpath of $P$ between the vertices $u_p$ and $u_q$,
	which has some integer length $\ell \in \mathbb{N}_0$.
Hence $P$ has length $\delta' = \ell + \lambda_p + \lambda_q$.

Let $P' \subseteq P(G)$ be the path between the new vertices $p_{\varepsilon}$ and $q_{\varepsilon}$ consisting of $Q$ and $\{ p(u_p,v_p,\lambda) \mid 0 \leq \lambda \leq \lambda_p + x\varepsilon \}$
	and $\{ p(u_q,v_q,\lambda) \mid 0 \leq \lambda \leq \lambda_q + y\varepsilon \}$.
Path $P'$ has minimum length of $p_{\varepsilon},q_{\varepsilon}$-paths that visit vertices $u_p$ and $u_q$.
Further, $P'$ has length $\ell+\lambda_p+\lambda_q+\varepsilon = \delta' + \varepsilon$.
		
Assume, for the sake of contradiction,
	that there is a path between $p_\varepsilon$ and $q_\varepsilon$ of length $<\delta'+\varepsilon$.
Then this path may not visit $u_p$ and $u_q$.
Further, there is an $\varepsilon'' \in [0,\varepsilon)$
	such that there is a shortest path $P''$ between $p_{\varepsilon''}$ and $q_{\varepsilon''}$
	using not both of $u_p$ and $u_q$ and of length \emph{exactly} $\delta' + \varepsilon''$.
To see this, note that the distance of a shortest path between $p_{\varepsilon}$ and $q_{\varepsilon}$ that does not use both vertices $u_p$ and $u_q$ is a continuous function in $\varepsilon$.
By symmetry assume that $P''$ does not use $v_p$.
\begin{itemize}
\item
First, consider that $P''$ uses vertices $v_p$ and $u_q$.
Consider that $d(u_p,u_q) = d(v_p,u_q)$.
Let path $P^\star$, analogously to $P'$,
	consist of $Q$,
	$\{ p(u_p,v_p,\lambda) \mid 0 \leq \lambda \leq \lambda_p + x{\varepsilon''} \}$
	and $\{ p(u_q,v_q,\lambda) \mid 0 \leq \lambda \leq \lambda_q + y{\varepsilon''} \}$.
Then each of $P''$ and $P^\star$ has length $\delta' + \varepsilon''$.
Since $d(u_p,u_q) = d(v_p,u_q)$, point $p_{\varepsilon''}$ has to have equal distance to $u_p$ and to $u_q$.
Thus, $p_{\varepsilon''}$ is the half-integral point $p(u_p,v_p,\frac{1}{2})$.
A contradiction to that no point $p(u_p,v_p,\lambda_p + x \varepsilon')$ for $0\leq \varepsilon' < \varepsilon''$ is half-integral.

If $d(u_p,u_q) \neq d(v_p,u_q)$, the distances $d(u_p,u_q)$ and $d(v_p,u_q)$ differ by at least one.
Then by similar arguments it follows that $p_{\varepsilon''}=p(u_p,v_p,0)$ or $p_{\varepsilon''}=p(u_p,v_p,1)$, hence is half-integral.
Again a contradiction to that no point $p(u_p,v_p,\lambda_p + x \varepsilon')$ for $0\leq \varepsilon' < \varepsilon''$ is half-integral.
\item
Now, consider that path $P''$ uses vertices $v_p$ and $v_q$.
Assume that $P$ and $P''$ intersect in an inner point.
Then, they in particular intersect in a point that is at a vertex $u \notin \{u_p,v_p\}$.
Analogously to the previous case, this implies that $p_{\varepsilon''}$ is half-integral.
Thus $P$ and $P''$ must have inner-points disjoint.
Then $P \cup P'' \subseteq P(G)$ forms a cycle of length $2(\delta' + \varepsilon'')$ which is an integer.
Hence $\delta'+\varepsilon''$ is half-integral.
A contradiction to that no $\delta' + \varepsilon'$ for $\varepsilon' < \varepsilon$ is half-integral.
\end{itemize}
Thus, $P'$ is indeed a shortest path between $p_\varepsilon$ and $q_\varepsilon$,
	which has length $\delta'+\varepsilon$, as desired.
\end{proof}

To properly prove \autoref{lemma:stay:critical}, we recall the technical framework.

An $(S,\delta)$-\define{pivot}, or simply a pivot, is a half-integral point $r \in P(G)$
	where two points $p,q \in S$, the witnesses, have distances $d(p,r) = d(r,q) = \tfrac{\delta}{2}$.
Since $S$ is $\delta$-dispersed, there is a shortest $p,q$-path of length $\delta$ containing $r$.
Let $\pivots{S}{\delta}$ be the set of $(S,\delta)$-pivots,
	and let $W(S,\delta) \subseteq \binom{S}{2}$ be the family of pairs of points from $S$, that witness some $(S,\delta)$-pivot.

We construct an \define{auxiliary graph} $G_S$ on vertex set $S \cup \pivots{S}{\delta}$:
\begin{itemize} 
	\item
	For $\{p,q\} \in W(S,\delta)$ and for every point $r \in \pivots{S}{\delta}$ they witness,
		add edges $\{p,r\},\{r,q\}$; and
	\item
	for every pair of points $\{p,q\} \in \binom{S}{2} \setminus W(S,\delta)$ with distance $d(p,q)=\delta$, add edge $\{p,q\}$.
\end{itemize}
Note that, for every edge $\{r,p_1\}$ with $r \in \pivots{S}{\delta}$, there is at least one other edge $\{r,p_2\}$ such that $p_1,p_2$ witness pivot $r$.

A path $P=(p_0,p_1,\dots,p_s),s \geq 1$ in graph $G_S$ is a \define{spine} if points $p_1,\dots,p_{s}$ are not half-integral.
Note that any sub-sequence $(p_0,\dots,p_i)$, $1\leq i \leq s$ is also a spine.

\medskip

Consider a spine $P=(p_0,\dots,p_s)$.
We define its half-integral \define{velocities} $\vel_P: \{p_0,\dots,p_s\} \allowbreak \to \HHH$
	depending on \define{signs} $\sgn_P: \{p_1,\dots,p_s\} \to \{-1,1\}$,
	which in turn depend on $\change_P: \{p_0,\dots,p_s\} \to \{-1,1\}$.
We may drop the subscript $P$, if it is clear from the context.
Let $\vel(p_0)=0$.
Let $\vel(p_1) = \frac{1}{2}$, if $p_0 \in \pivots{S}{\delta}$, and let $\vel(p_1)=1$, if $p_0 \in S$.
For $i \geq 1$, let
$$ \vel(p_{i+1}) \; \coloneqq \; \vel(p_i)+ \sgn(p_{i+1}) . $$ 
Thus, $\sgn \in \{-1,1\}$ indicates whether the velocity increases or decreases.
The current $\sgn$ is unchanged unless $\change$ is negative.
Let $\sgn(p_1) = 1$.
For $2 \leq i \leq s$, let
	$$
	\sgn(p_{i}) \; \coloneqq \; \change(p_{i-1}) \sgn(p_{i-1}) \; = \; \prod_{0< j< i} \change(p_j) , $$ 
where for $1 \leq i \leq s-1$,
$\change(p_i) \in \{1,-1\}$, and $\change(p_i)=1$ if and only if $ \ori{p_{i-1}}{p_i} \neq \ori{p_{i+1}}{p_i}$. 
	$$ 
		\change(p_i) \coloneqq
		\begin{cases}
			\phantom{-}1, & \ori{p_{i-1}}{p_i} \neq \ori{p_{i+1}}{p_i}, \\ 
			-1, & \text{else}.
		\end{cases}
	$$
Value $\change(p_i)$ describes whether a shortest path between $p_{i-1}$ and $p_i$
	and a shortest path between $p_{i}$ and $p_{i+1}$ overlap.
If $\change(p_i)=1$, they do not overlap.
If $\change(p_i)=-1$, the paths overlap;
	and this will cause a flip for the sign, i.e., whether the speed is increasing or decreasing.

\medskip

Let $\varepsilon \geq 0$ and consider a spine $P = (p_0,\dots,p_i)$.
We move all points $p_i$ for $1\leq i \leq s$ by $\sgn_P(p_i) \vel_P(p_i) \varepsilon$ away from their predecessor $p_{i-1}$.
Specifically, when $\lambda_i$ is such that $p_{i} = p(\cdot, \nori{p_{i-1}}{p_{i}}, \lambda_i)$,
	we define the moved point as
\begin{equation*}
	\label{eq:moved:point}
	(p_i)_{P,\varepsilon} \coloneqq p \big(\cdot, \; \nori{p_{i-1}}{p_{i}}, \; \lambda_i + \sgn_P(p_i) \vel_P(p_i) \varepsilon \big) .
\end{equation*}

Doing so, points $(p_i)_{P,\varepsilon},(p_{i+1})_{P,\varepsilon}$ maintain a distance of $\delta+\varepsilon$,
	under the condition that we do not encounter a half-integral point in the process
	and that $\delta+\varepsilon$ itself does not become half-integral.

\newcommand{\lemmaTextMovingPointsPreview}{
	Let $P = (p_0,\dots,p_i)$ be a spine.
	For $ \varepsilon' \in (0,\varepsilon)$, let neither $(p_{i-1})_{P,\varepsilon'}$ nor $(p_i)_{P,\varepsilon'}$
		nor $\delta' + \varepsilon'$ be half-integral.
	Then points $\{(p_{i-1})_{P,\varepsilon}, (p_{i})_{P,\varepsilon}\}$ are $(\delta + \varepsilon)$-critical.
}

\begin{lemma}
\label{lemma:moving:points:preview}
\label{lemma:appendix:moving:points:preview}
\lemmaTextMovingPointsPreview
\end{lemma}
\begin{proof}
For simplicity, let us write $(p_{i'})_\varepsilon$ instead of $(p_{i'})_{P,\varepsilon}$
	for $0 \leq i' \leq i$.

If $i=1$, then $\vel(p_{0})=0$, and $p_{0} = (p_{0})_\varepsilon$ by definition.
Further $\vel(p_1) = 1$, since $p_0,p_1 \in S$.
Also $\sgn(p_1)=1$.
Let $p_1 = p(\cdot, \nori{p_{0}}{p_1}, \mu)$.
Then $(p_1)_\varepsilon = p(\cdot, \nori{p_{0}}{p_1}, \mu + \sgn(p_1) \vel(p_1) \varepsilon) = p(\cdot,\nori{p_{0}}{p_1},\mu+\varepsilon)$.
There is a shortest path between $p_0,p_1$ that uses $\ori{p_0}{p_1}$.
Hence \autoref{lemma:rounding:one} applies to old points $p_{0},p_{1}$ and new points $(p_{0})_{\varepsilon}, (p_{1})_\varepsilon$
	with $x = 0$ and $y = 1$.
It follows that $d(p_\varepsilon,q_\varepsilon) = \delta + \varepsilon$, as desired.
		
Now consider that $i> 1$, and hence point $p_{i-1}$ is not half-integral.
Then there is an edge position $\lambda \in (-\tfrac{1}{2},\tfrac{1}{2})$ such that
\begin{align*}
	p_{i-1} &= p \big(\cdot, \; \nori{p_{i-2}}{p_{i-1}}, \; \tfrac{1}{2} + \lambda \big) & \\
	&= p \big(\cdot, \; \nori{p_{i}}{p_{i-1}}, \; \tfrac{1}{2} -\change(p_{i-1}) \lambda \big) & 
\end{align*}
For the new point $(p_{i-1})_\varepsilon$ there is some edge position $\lambda^\star = \lambda + \sgn(p_{i-1})  \vel(p_{i-1})\varepsilon \in [-\tfrac{1}{2}, \tfrac{1}{2}]$
	such that $(p_{i-1})_\varepsilon = p(\cdot, \; \nori{p_{i-2}}{p_{i-1}}, \; \tfrac{1}{2} + \lambda^\star )$.
Indeed, the new edge position may be integral.
Still $\nori{p_{i-2}}{p_{i-1}}$ is defined, since it relies on the old point $p_{i-1}$.
By \autoref{lemma:dir:change:flip},
		\begin{align*}
				(p_{i-1})_\varepsilon &= p \big(\cdot, \; \nori{p_{i}}{p_{i-1}}, \; \tfrac{1}{2} - \change(p_{i-1}) \lambda^\star) \\
				&= p \big(\cdot, \; \nori{p_{i}}{p_{i-1}}, \; \tfrac{1}{2} -\change(p_{i-1}) \lambda - \change(p_{i-1}) \sgn(p_{i-1})  \velo(p_{i-1})\varepsilon \big) \\
				&= p \big(\cdot, \; \nori{p_i}{p_{i-1}}, \; \tfrac{1}{2} -\change(p_{i-1}) \lambda - \sgn(p_i)  \velo(p_{i-1})\varepsilon \big) .
		\end{align*}
		
		Let edge position $\mu \in (-\tfrac{1}{2},\tfrac{1}{2})$ be such that $p_i = p(\cdot, \nori{p_{i-1}}{p_i}, \tfrac{1}{2} + \mu)$.
		Then
		\begin{align*}
				(p_i)_\varepsilon &= p \big( \cdot, \; \nori{p_{i-1}}{p_i}, \; \tfrac{1}{2} + \mu + \sgn(p_i) \vel(p_i) \varepsilon \big) \\
				&= p \big( \cdot, \; \nori{p_{i-1}}{p_i}, \; \tfrac{1}{2} + \mu + \sgn(p_i) (\vel(p_{i-1}) + \sgn(p_{i})) \varepsilon \big) \\
				&= p \big( \cdot, \; \nori{p_{i-1}}{p_i}, \; \tfrac{1}{2} + \mu + ( \sgn(p_i) \vel(p_{i-1}) + 1)) \varepsilon \big) .
		\end{align*}
There is a shortest path between $p_{i-1},p_i$ that does not use $\nori{p_{i-1}}{p_i}$ and $\nori{p_i}{p_{i-1}}$.	
Thus \autoref{lemma:rounding:one} applies to points $(p_{i-1})_\varepsilon, (p_{i})_\varepsilon$
	with $x = -\sgn(p_i)\velo(p_{i-1})$ and $y = \sgn(p_i)\velo(p_{i-1})+1$.
It follows that $d(p_\varepsilon,q_\varepsilon) = \delta + \varepsilon$, as desired.
\end{proof}

\newcommand{\lemmaTextPathFromPGtoG}{
Let $P \subseteq P(G)$ be a path of points of length $\ell \in \mathbb{R}^+$.
Then $G$ contains a path of length at least $\lfloor \ell-1 \rfloor$.
}

\begin{lemma}
\label{lemma:appendix:path:from:PG:to:G}
\label{lemma:path:from:PG:to:G}
\lemmaTextPathFromPGtoG
\end{lemma}
\begin{proof}
Let $\{u_1,u_2\},\{u_2,u_3\},\dots,\{u_s,u_{s+1}\}$
	be the list of edges visited by \stress{non-integral} points when we follow the path $P$ of length $\ell$ from start to finish.
Beside possibly $\{u_1,u_2\}=\{u_s,u_{s+1}\}$, no edge may be repeated.
Hence also $\lceil \ell \rceil \geq s$.
Particularly none of the vertices $u_2,u_3, \dots ,u_s$ occurs twice in $P$.
Thus $u_2,u_3, \dots ,u_s$ is a path in $G$ consisting of at least $\lfloor \ell-1 \rfloor$ edges.
\end{proof}

\begin{lemma}
\label{lemma:not:half:integral}
\label{lemma:appendix:not:half:integral}
\lemmaTextDeltaNotHalfIntegral
\end{lemma}
\begin{proof}
Assume, for the sake of contradiction, that $\delta$ is half-integral,
	hence that $\delta = a$ or $\delta = \tfrac{a}{2}$ for some $a \in \mathbb{N}$.
Recall that $a > 2L+2$.
Then $\delta \geq \tfrac{a}{2} > L$.
Recall that $|S|>2$.
Let $P \subseteq P(G)$ be a shortest path between distinct points $p,q \in S$,
	which is of length $\delta > L$.
Then $G$ contains a path of length at least $L$,
	as seen in \autoref{lemma:path:from:PG:to:G}.
Contradiction.
\end{proof}

Recall the definition of the root points $R$.
Let $R_0$ be the set of half-integral points in $G_S$.
There may be some components of the auxiliary graph $G_S$ without a point in $R_0$.
Let $R$ result from $R_0$ by adding exactly one point from every component that has no point in $R_0$.
Now, when we restrict spines to those that begin with a point from $R$,
	the implied movement of points is uniquely determined.

For each point~$p$, the new point is $p_\varepsilon$,
	and the resulting set is $S_\varepsilon \coloneqq \{\, p_\varepsilon \mid p \in S \,\}$. 
For a point ${p_0 \in R}$, simply $(p_0)_\varepsilon \coloneqq p_0$.
For other points $p_i \in S \setminus R$,
	hence that are the $i$-th point, for some $i \geq 1$, of some spine $P = (p_0,\dots,p_i)$ with start $p_0 \in R$,
	point $(p_i)_\varepsilon$ is defined as for the spine $P$,
	formally $(p_i)_\varepsilon \coloneqq (p_i)_{P,\varepsilon}$.
By \autoref{lemma:equal:potential}, the definition of $p_\varepsilon$ is independent of the considered spine.

We further recall $\me$ and its defining events.
Let $\me \geq 0$ be minimum such that $\delta+\me = \delta^\star$
	or one of the following events occurs:
	\begin{itemize}
		\item
		(\labeltext{Event 1}{event:1:a})
	For a $\delta$-uncritical pair of points $\{p,q\} \in \binom{S}{2}$,
	now $\{p_\me,q_\me\}$ $(\delta+\varepsilon)$-critical.
		\item
		(\labeltext{Event 2}{event:2:a})
		a non-half-integral $p \in S$, now $p_\me$ half-integral, or
		\item
		(\labeltext{Event 3}{event:3:a})
		for $r \in P(G) \setminus \pivots{S}{\delta}$, now $r \in \pivots{S_\me}{\delta+\me}$.
	\end{itemize}
We claim that the minimum $\me \geq 0$ is defined.
This is due to that the above events depend on continuous functions in $\varepsilon$,
	which are the distance of $p_{\varepsilon}$ to its closest half-integral point, and the distance between points $p_\varepsilon$ and $q_{\varepsilon}$ for $p,q \in S$.

\begin{lemma}[(\autoref{lemma:stay:critical} restated)]
\label{lemma:appendix:stay:critical}
\lemmaTextStayCritical
\end{lemma}
\begin{proof}
This is an immediate consequence from \autoref{lemma:moving:points:preview}.
Because of \autoref{lemma:not:half:integral} $\delta + \varepsilon'$ is not half-integral for $\varepsilon' \in (0,\varepsilon)$ is $\varepsilon' \in $.
That event 2 did not occur before implies that neither $(p_i)_{P,\varepsilon'}$ nor $(p_{i+1})_{P,\varepsilon'}$ is half-integral for $\varepsilon' \in (0,\varepsilon)$ is $\varepsilon' \in $.
\end{proof}

\subsection{Proof of \autoref{lemma:rounding:invariant}}

We delay the (lengthy and technical) proof of \autoref{lemma:equal:potential} to the next subsection.
Here we prove the monotonicity result, \autoref{lemma:rounding:invariant}.

\begin{lemma}[(\autoref{lemma:rounding:invariant} restated)]
	\label{lemma:appendix:rounding:invariant}
	\lemmaTextRoundingInvariant
\end{lemma}
\begin{proof}
To show claims $(1)$, $(2)$, $(3)$, we prove the following more convenient version:	
\begin{itemize}
\item $(1')$
For every $\{p,q\} \in \binom{S}{2} \setminus (E(G_S) \cup W(S,\delta))$
	we have $d(p_\me,q_\me) \geq \delta + \me$.
\item $(3')$
For every $\{p,q\} \in \binom{S}{2} \cap W(S,\delta)$ that witness some pivot $r \in \pivots{S}{\delta}$,
	we have $\{p_\me,q_\me\} \in W(S_{\me},\delta + \me)$ and $r \in \pivots{S_\me}{\delta+\me}$.
Moreover, for every pivot $r \in \pivots{S}{\delta}$ with edge $\{r,p\} \in E(G_S)$,
	we have $\{r, p_\me\} \in E(G_{S_\me})$.
\item $(2')$
For every $\{p,q\} \in \binom{S}{2} \cap E(G_S)$,
	we have $d(p_\me,q_\me) = \delta + \me$.
\end{itemize}
Here $(1')$, $(3')$, $(2')$ is a case-distinction for $(1)$.
Statement $(3)$ follows from $(3')$ since any pivot $r$ has some witnesses $p,q$.
Claim $(2)$ for edges in $G_S$ involving a pivot is shown by $(3')$.
Claim $(2)$ for the remaining edges $\{p,q\} \in \binom{S}{2} \cap E(G_S)$ follows from $(2')$.
So it remains to show the above claims.
\begin{itemize}
\item $(1')$
Assume, for the sake of contradiction,
	that there are points $\{p,q\} \in \binom{S}{2} \setminus (E(G_S) \cup W(S,\delta))$ where $d(p_\me,q_\me) < \delta + \me$.
Then there is an $\varepsilon' \in [0,\me)$ such that points $p_{\varepsilon'},q_{\varepsilon'}$ have distance exactly $\delta +\varepsilon'$.
Then event~\ref{event:1:n} should have triggered for $\varepsilon'$ and points $p_{\varepsilon'},q_{\varepsilon'}$.
A contradiction to the minimality of $\me$.
\item $(3')$
Consider points $\{p,q\} \in \binom{S}{2} \cap W(S,\delta)$ that witness some pivot $r \in \pivots{S}{\delta}$.
Especially $p,q$ may not be half-integral due to \autoref{lemma:not:half:integral}.
Hence there are edge positions $\lambda_p,\lambda_q \in (0,1)$
	such that $p = p(\cdot, {\nori{r}{p}}, \lambda_p)$
	and $q = p(\cdot, \nori{r}{q}, \lambda_q)$.
Further, the two-element sequences $P = (r,p)$ and $Q =(r,q)$ are spines.
By construction, we have that $\vel_P(p) = \vel_Q(q) = \frac{1}{2}$.
It follows that the new points are at $p_\me = p(\cdot,\nori{r}{p},\lambda_p + \frac{1}{2}\me)$ and $q_\me = p(\cdot,{\nori{r}{q}},\lambda_q + \frac{1}{2}\me)$.
Then \autoref{lemma:rounding:one} applies to points $p_\me,q_\me$ with $\delta' = \delta$, $\varepsilon' = \me$
	and with $x+y=\tfrac{1}{2}+\tfrac{1}{2}=1$,
	since $\delta + \me \leq \delta^\star$ and \ref{event:2:n} has not occurred.
It follows that $d(p_\me, q_\me) = \delta+\me$, as desired.
Further points $p_\me, q_\me$ still witness $r$ as a $(\delta + \me)$-pivot.
In other words $r \in \pivots{S_\me}{\delta + \me}$.

Finally, consider pivot $r \in \pivots{S}{\delta}$ of $S$ with an edge $\{r,p\} \in E(G_S)$.
Then there is $q \in S$ such that $p,q$ witness pivot $r$.
As just observed, then $d(p_\me,r) = \tfrac{\delta+\me}{2}$
	and hence $\{p_\me,r\} \in E(G_{S_\me})$.

\item $(2')$
Consider points $\{p,q\} \in \binom{S}{2} \cap E(G_S)$.
We will show that $d(p_\me,q_\me) = \delta + \me$.
There is a spine $P = (p_0,\dots,p_{i})$ with $p_0 \in R$ and $\{p_{i-1},p_i\}=\{p,q\}$.
By symmetry, we may assume that $p_{i-1}=p$ and $p_i=q$.
Since event~\ref{event:2:n} did not occur for $\varepsilon' \in (0,\varepsilon)$
	\autoref{lemma:moving:points:preview} applies.
It follows that $d(p_\me,q_\me) = \delta + \me$.
\end{itemize}
This finishes the proof.
\end{proof}

\subsection{Proof of \autoref{lemma:equal:potential}}
\label{section:proof:of:equal:pot}

Now we work towards proving \autoref{lemma:equal:potential}.

\begin{lemma}[(\autoref{lemma:equal:potential} restated)]
\label{lemma:appendix:equal:potential}
{\textLemmaEqualPotential}
\end{lemma}

We begin by explicitly stating the edge position of a point $p_i$ of a spine $P=(p_0,\dots,p_i)$ depending on $\vel(p_i)$, $\sgn(p_i)$
	and the edge position $\lambda_0$ of $p_0$.
Then our final proof is by contradiction. It leads the assumption that different spines that end in the same point do not agree on $\vel$ and $\sgn$, such that there must be a path $G$ longer than~$L$,
	a contradiction.

\medskip

Consider a point $p_i$ that is part of a spine $P=(p_0,\dots,p_i)$.
We will state its edge position~$\mu_i$ depending on its velocity~$\vel(p_i)$, its sign~$\sgn(p_i)$ and the edge position~$\lambda_0$ of~$p_0$. 
We orientate the edge position~$\mu_i$ according to $\nori{p_{i-1}}{p_i}$,
	its direction to the predecessor point $p_{i-1}$.
That is, let $\mu_i$ be such that $p_i=p( \cdot, \nori{p_{i-1}}{p_i}, \mu_i)$;
	hence $\mu_i$ is a part of the length of any shortest path between~$p_i$ and~$p_{i-1}$,
	specifically the part using the edge of $p_i$ (assuming $q$ is on another edge).

We give some intuition:
Assume that we already have some expression for the edge position $\mu_{i-1}$ of the previous point $p_{i-1}$,
	hence we know the value $\mu_{i-1} \in (0,1)$ such that $p_{i-1} = p(\cdot, \nori{p_{i-2}}{p_{i-1}}, \mu_{i-1})$.
Consider that a shortest $p_{i-1},p_i$-path starts from $p_{i-1}$ in the direction of vertex $\nori{p_{i-2}}{p_{i-1}}$,
	in other words $\nori{p_{i-2}}{p_{i-1}} \neq \nori{p_{i}}{p_{i-1}}$ and by definition $\change(p_{i-1})=1$.
Then it is not difficult to follow that $\mu_i = \fracp(\mu_{i-1} + \delta)$ where $ \fracp(y) = y - \lfloor y \rfloor$ is the \define{fractional part} of $y \in \mathbb{R}$. 
If however $\change(p_{i-1})=-1$, hence $\nori{p_{i-2}}{p_{i-1}} = \nori{p_{i}}{p_{i-1}}$,
	then point~$p_{i}$ has the same edge position as point $p_{i-2}$ but `flipped'.
More precisely, when $\mu_{i-1} \in (0,1)$ is such that $p_{i-2} = p(\cdot, \nori{p_{i-3}}{p_{i-2}}, \mu_{i-2})$,
	we have $p_{i} = p(\cdot, \nori{p_{i-1}}{p_{i}}, 1 - \mu_{i-2})$.
	
To handle this flip of the edge position more conveniently,
	let us specify the edge positions relatively to the middle $\tfrac{1}{2}$.
That is, we are interested in the (shifted) edge position $\lambda_i \in (-\tfrac{1}{2},\tfrac{1}{2})$
	such that $p_i = p(\cdot, \nori{p_{i-1}}{p_{i}}, \tfrac{1}{2} + \lambda_i)$ for $i\geq 1$.
It follows that:

\begin{lemma} 
\label{lemma:dir:change:flip}
For a spine $P=(p_0,\dots,p_{i+1})$, $i\geq 1$,
	there is a $\lambda \in (-\tfrac{1}{2},\tfrac{1}{2})$
	such that $p(\cdot,\nori{p_{i-1}}{p_{i}},\tfrac{1}{2}+\lambda)
		= p_{i} = p(\cdot,\nori{p_{i+1}}{p_{i}},\tfrac{1}{2} - \change(p_{i-1}) \lambda) $.
\end{lemma}
As noted earlier, if the predecessor point $p_{i-1}$ has $\change(p_{i-1})=-1$,
	the sign of the edge position of $p_i$ should flip.
This sign change is tracked by $\sgn$.
In turn $\vel(p_i)$ sums up the values of $\sgn$ of points $p_1,\dots,p_i$,
	and hence allows to state the edge position relatively to that of $p_0$.
For $\lambda \in (-\tfrac{1}{2},\tfrac{1}{2})$ and $x \in \HHH$, let
$$
	\lf{\lambda}{x} \coloneqq \fracp \left(\tfrac{1}{2} + \lambda + x\delta \right) - \tfrac{1}{2}.
$$	
Consider that $\lambda$ 
	is such that $p_0 = p(\cdot, \, \ori{p_{1}}{p_0} , \, \frac{1}{2}+\lambda)$,
	assuming for now that $p_0$ is non-half-integral.
Then $\lf{\lambda}{\vel(p_i)}$ is the distance of the edge position of $p_i$ to the middle $\tfrac{1}{2}$,
	in other words, the edge position of $p_i$ up to the sign.
	
Later we rely on the following observation:

\begin{lemma} 
\label{lemma:lambda:function:composition}
	$  \lf{\lambda}{x+y} = \lf{ \lf{\lambda}{x} }{y} $,
	for every $\lambda  \in [-\tfrac{1}{2}, \tfrac{1}{2})$ and $x,y \in \HHH$.
\end{lemma}

Further, let us state the edge position even for half-integral points,
	as for example it may be the case for $p_0$.
To do so, we extend the definition of $\ori{p}{q}$ to also apply for half-integral $p$:
For points $p,q \in P(G)$, let $\orie{q}{p}$ be equal to $\ori{p}{q}$ when $\ori{p}{q}$ is defined.
Else, $p$ must be half-integral as observed in \autoref{lemma:rounding:orientation}.
If $p$ is such that $p=(u,v,\tfrac{1}{2})$, fix $\orie{q}{p}$ to be $u$ or $v$ arbitrarily,
	and let $\norie{p}{q}$ be the unique vertex in $\{u,v\} \setminus \{\orie{p}{q}\}$.
Else, if $p$ is such that $p=p(u,v,0)$ for some edge $\{u,v\}\in E(G)$,
	let $\orie{q}{p} = u$ and let $\norie{q}{p} = v'$ where $v'$ is an arbitrary neighbor of $v$ in $G$.
As a result, we may for example define $\lambda_0 \in [-\tfrac{1}{2},\tfrac{1}{2})$
	such that $p_0 = p(\cdot, \norie{q}{p}, \tfrac{1}{2}+\lambda_0)$.
Here again the missing entry is irrelevant for $\lambda_0$.

Moreover, let us also extend the definition of a spine to allow a final half-integral point.
A path $P=(p_0,\dots,p_s), s \geq 1$ in graph $G_S$ is a \define{generalized spine}
	if $(p_0,\dots,p_{s-1})$ is a spine.
Values $\vel$, $\sgn$ and $\change$ are defined analogously as we did for spines (still based on $\noripara$).
Also $\vel$ and $\sgn$ are defined for $p_s$.
Only $\change(p_s)$ is undefined which however is not used.

\begin{lemma}\label{lemma:lambda:jumps}
Let $P = (p_0,\dots,p_s)$ be a generalized spine.
Let $\lambda \in [-\tfrac{1}{2},\tfrac{1}{2})$ be such that 
	 $p_0= p(\cdot, \, \orie{p_{1}}{p_0} , \, \frac{1}{2}+\lambda)$.
Then, for $i\in\{1,\dots,s\}$,
	$$ p_i \; = \; p \big(\cdot, \; \norie{p_{i-1}}{p_i}, \; \tfrac{1}{2} + \sgn(p_i) \cdot \lf{\lambda}{\vel(p_i)} \big) . $$
\end{lemma}
\begin{proof}
Before we prove the actual statement,
	let us consider how the edge positions of some non-half-integral points $p,q$ in distance $\delta'\in \{\tfrac{\delta}{2}, \delta\}$ relate.
Let $\mu_p \in [0,1)$ be such that $p = p(\cdot, \orie{q}{p},\mu_p)$,
	and let $\mu_q \in [0,1)$ be such that $q = p(\cdot, \norie{p}{q}, \mu_q)$.
If $p,q$ are on the same edge, then $\mu_q = \mu_p + \delta'$.
If $p,q$ are not on the same edge, there is a shortest path $P$ between $p$ and $q$ that uses vertices $\orie{q}{p}$ and $\orie{p}{q}$.
Let integer $\ell \geq 0$ be the length of path $P$ restricted to the path between $\orie{q}{p}$ and $\orie{p}{q}$.
Then $\delta' = (1-\mu_p) + \ell + \mu_q$.
Hence
$ \mu_p + \delta' = \ell + \mu_q + 1$,
	which implies $\lfloor \mu_p + \delta' \rfloor = \ell + 1$,
	since $\ell,1$ are integer and $\lfloor \mu_q \rfloor = 0$.
Then
\begin{equation}\label{eq:p:q:jump:raw}
\mu_q \; = \;
	\mu_p  + \delta' - \big\lfloor \mu_p + \delta' \big\rfloor .
\end{equation}
This also applies when $p,q$ are on the same edge, where we followed $\mu_q = \mu_p + \delta' \in [0,1)$.
Let us state the edge positions relatively to the middle $\tfrac{1}{2}$.
Let $\lambda_p \in \left[-\tfrac{1}{2},\tfrac{1}{2}\right)$ be such that
	$p = p(\cdot, \orie{q}{p}, \tfrac{1}{2} + \lambda_p)$,
	and let $\lambda_q \in \left[-\tfrac{1}{2},\tfrac{1}{2}\right)$ be such that
	$q = p(\cdot, \norie{p}{q}, \tfrac{1}{2} + \lambda_q)$.
That means $\tfrac{1}{2} + \lambda_p = \mu_p$ and $\tfrac{1}{2} + \lambda_q = \mu_q$.
It follows that
\begin{equation}\label{eq:p:q:jump}
	\lambda_q \; = \; \lambda_p + \delta' - \left\lfloor \tfrac{1}{2} + \lambda_p + \delta' \right\rfloor
	\;=\; \lf{\lambda_p}{\tfrac{\delta'}{\delta}}
	.
\end{equation}
Now, we are ready to prove the statement that
$$ p_i \; = \; p \big(\cdot, \; \norie{p_{i-1}}{p_i}, \allowbreak \; \tfrac{1}{2} + \sgn(p_i) \cdot \lf{\lambda}{\vel(p_i)} \big) $$
by induction on $i\in\{1,\dots,s\}$.

\medskip

Induction base, $i=1$:
Let $\delta'$ be the distance between $p_0,p_1$ which is either $\delta$ or $\tfrac{\delta}{2}$.
Then $p_1$ has velocity $\vel(p_1) = \tfrac{\delta'}{\delta}$. 
Note that \autoref{eq:p:q:jump} applies for points $p_0=p, p_1=q$ and $\lambda_p = \lambda$.
It follows that $p_1 = p(\cdot, \norie{p_0}{p_1}, \tfrac{1}{2}+ \lf{\lambda}{ \vel(p_1) } )$.
Since $\sgn(p_1)=1$, we are done.

\medskip
	
Induction step, $i \leadsto i+1$, $i\geq1 $:
Note that point $p_i$ is non-half-integral.
Let $z \coloneqq \lf{\lambda}{\vel(p_i)}$.
By the induction hypothesis,
\begin{align*}
	{p_i} \; =& \;\; p\big( \cdot, \norie{p_{i-1}}{p_i}, \tfrac{1}{2} + \sgn(p_i) z \big) && \\
	\; =& \;\; p\big( \cdot, \nori{p_{i-1}}{p_i}, \tfrac{1}{2} + \sgn(p_i) z \big), && \mid \text{$p_i$ is non-half-integral} \\
	=& \;\; p\big(\cdot, \ori{p_{i+1}}{p_{i}}, \tfrac{1}{2} + \change(p_i) \sgn(p_i) z \big), && \mid \text{\autoref{lemma:dir:change:flip}} \\
	=& \;\; p\big(\cdot, \orie{p_{i+1}}{p_{i}}, \tfrac{1}{2} + \change(p_i) \sgn(p_i) z \big) && \\
	=& \;\; p\big(\cdot, \orie{p_{i+1}}{p_{i}}, \tfrac{1}{2} + \sgn(p_{i+1}) z \big).
\end{align*}
Hence, the edge position of point~$p_i$ is $\lambda_{i} = \sgn(p_{i+1}) z$ when defined as above.
Let value $\lambda_{{i+1}} \in [-\tfrac{1}{2},\tfrac{1}{2})$ be such that $p_{i+1} = p(\cdot, \norie{p_{i}}{p_{i+1}},\frac{1}{2} + \lambda_{i+1})$.
Then \autoref{eq:p:q:jump} applies for points $p=p_{i}$ and $q=p_{i+1}$ with edge positions $\lambda_i$ and $\lambda_{i+1}$, respectively, and distance between them of~$\delta' = \delta$.
We define $\lambda_{i+1}' \coloneqq \sgn(p_{i+1}) \lambda_{i+1}$.
Thus we have to show that $\lambda_{i+1}' = \lf{\lambda}{\vel(p_{i+1})}$.
We have that $\lambda_{i+1}' = \sgn(p_{i+1}) \cdot \lf{\lambda_i}{1}$, according to \autoref{eq:p:q:jump}.

First, consider that $\sgn(p_{i+1}) = 1$, and hence $\vel(p_{i+1})=\vel(p_i)+1$.
Then
\begin{align*}
	\lambda_{i+1}'
	= \lf{\lambda_i}{1}
	= \lf{ \lf{\lambda}{\vel(p_i)} }{1}
	\overset{Obs.~\ref{lemma:lambda:function:composition}}{=}
	\lf{ \lambda }{\vel(p_i)+1}
	= \lf{ \lambda }{\vel(p_{i+1})},
\end{align*}
as desired.
Now, consider that $\sgn(p_{i+1}) = -1$, and hence $\vel(p_{i+1})=\vel(p_i)-1$.
We use the fact that $- \lfloor x \rfloor = \lfloor -x+1 \rfloor $ for non-integral $x$.
Then
\begin{align*}
	\lambda_{i+1}'
	&= - \lf{\lambda_i}{1} && \\
	&= - \big( \lambda_{i} + \delta - \lfloor \tfrac{1}{2} + {\lambda}_{i} + \delta \rfloor \big) && \\
	&= z - \delta + \lfloor \tfrac{1}{2} -z + \delta \rfloor && \\
	&= z - \delta - \lfloor \tfrac{1}{2} + z - \delta \rfloor && \\
	&= \lf{ z }{-1} && \\
	&= \lf{\lambda}{\vel(p_i)-1} && \mid \text{\autoref{lemma:lambda:function:composition}}
\end{align*}
as desired.
\end{proof}

Now, we show that a generalized spine $(p_0,p_1,\dots,p_s)$ implies a path of length $\lfloor |\vel(p_s)|\delta \rfloor$ in $G$.
Roughly speaking, we can combine shortest paths between the pairs $(p_0,p_1)$, $(p_1,p_2),\dots$ to obtain a long path in $P(G)$, which in turn implies a long path in $G$.

\begin{lemma}
	\label{lemma:oriented:path}
	Let $P = (p_0,p_1,\dots,p_{s})$ be a generalized spine. 
	Then $G$ contains a path of length $\lfloor |\vel(p_s)| \delta -1 \rfloor$.
\end{lemma}
\begin{proof}
Before we begin, let us rule out the special case that $\vel(p_{i-1})=\frac{1}{2}$ and $\vel(p_{i})=-\frac{1}{2}$, $i \in \{1,\dots,s\}$,
	which would imply that $p_{i-1}$ and $p_{i}$ witness a pivot.

\begin{claim}
\label{lemma:no:pivot:potentials}
For no $i \in \{1,\dots,s\}$ we have $\vel(p_{i-1})=\frac{1}{2}$ and $\vel(p_{i})=-\frac{1}{2}$.
\end{claim}
\begin{claimproof}
Clearly we have $i\geq 2$ since $\vel(p_0)=0$.
There is a position $\lambda \in (-\tfrac{1}{2},\tfrac{1}{2})$
	such that $p_{i-1} = p(\cdot,\norie{p_{i-2}}{p_{i-1}}, \tfrac{1}{2} + \lambda)$
		and $p_{i} = p(\cdot,\norie{p_{i-1}}{p_{i}}, \tfrac{1}{2} - \lambda)$, according to \autoref{lemma:lambda:jumps}.
First assume that $\change(p_{i-1})=1$ hence that $\sgn(p_{i-1})=\sgn(p_{i})$.
Then we know that $p_{i-1} = p(\cdot,\norie{p_{i}}{p_{i-1}}, \tfrac{1}{2} - \lambda)$, as seen in \autoref{lemma:dir:change:flip}.
It follows that points $p_{i-1},p_i$ have integer distance $\tfrac{1}{2}+\lambda+\tfrac{1}{2}-\lambda+\ell = 1+\ell$ for some $\ell \in \mathbb{N}$,
	hence that $\delta$ is integer, a contradiction to \autoref{lemma:not:half:integral}.

Thus, it remains to consider that $\change(p_{i-1})=-1$ and hence that $\sgn(p_{i-1}) \neq \sgn(p_{i})$.
Then we know that $p_{i-1} = p(\cdot,\norie{p_{i}}{p_{i-1}}, \tfrac{1}{2} + \lambda)$, as seen in \autoref{lemma:dir:change:flip}.
That is, points $p_{i-1},p_i$ have some equal distance to a half-integral point $r \in P(G)$,
	which is midway on some path between the vertices $\orie{p_{i}}{p_{i-1}}$ and $\orie{p_{i-1}}{p_{i}}$
	(even when $p_{i-1}$ and $p_i$ are on the same edge).
Then $p_{i}, p_{i-1}$ witness the pivot $r$,
	in contradiction to that $\{p_{i}, p_{i-1}\} \in E(G_S)$.
\end{claimproof}

Further, we claim that it is safe to assume that $\vel(p_s)$ is positive and has maximum velocity among all points in $P$.
Clearly if $\vel(p_s)=0$ there is nothing to prove for the lemma.
Now consider that $\vel(p_s)<0$.
Then there is an index $i \in \{1,\dots,s-1\}$ where $\vel(p_i) \in \{\tfrac{1}{2},0\}$ and $\vel(p_j) < 0$ for every $j \in \{i+1,\dots,s\}$.
Particularly, $\vel(p_i)=0$ since otherwise it contradicts \autoref{lemma:no:pivot:potentials}. 
Consider the spine $P' = (p_i,\dots,p_s)$.
Observe that $\sgn_P(p_i)=-1 = -\sgn_{P'}(p_i)$
	and $\vel_P(p_i) = 0 = \vel_{P'}(p_i)$.
Further, $\change_P(p_j) = \change_{P'}(p_j)$ for $j \in \{i+1,\dots,s\}$.
Inductively, it follows that $\sgn_{P'}(p_j) = - \sgn_P(p_j)$
	and hence $\vel_{P'}(p_j) = - \vel_{P'}(p_j)$, for $j \in \{1,\dots,s\}$.
Thus, if $\vel_P(p_s)<0$, let us consider the spine $P'$ with non-negative $\vel_{P'}(p_s) = -\vel_P(p_s)$ instead of $P$.
Hence, it is save to only consider the case that $\vel(p_s) > 0$.
Now, if $\vel(p_i)\geq\vel(p_s)$, for some $i<s$, consider the sub-spine $(p_0,\dots,p_i)$ instead,
	which has the same $\change$, $\sgn$ and $\vel$ values.
Since $\vel(p_i) \geq \vel(p_s) > 0$ we have that $i\geq 1$.	
Hence, in the following, we may then assume that $z= \vel(p_s)$ is positive and 
has maximum velocity among the points of $P$.

\medskip

Let us define shortest paths $P_i \subseteq P(G)$ between $p_{i-1}$ and $p_i$ for $i \in \{1,\dots,s\}$.
For convenience, let $P_i \coloneqq \emptyset$ for $i \in \mathbb{N} \setminus \{1,\dots,s\}$.
For $i \in \{1,\dots,s\}$, let $P_{i}$ be some shortest path between $p_{i-1}$ and $p_{i}$
	with maximum intersection with the previous path $P_{i-1}$.
	Possibly their intersection is only $\{p_{i-1}\}$, in case that $\change(p_{i-1})=1$.
	Hence, the intersection of paths $P_{i-1}$ and $P_{i}$ forms a path starting at $p_{i-1}$.
It follows that $P_i \setminus P_{i-1}$ is a path.
	
For $j\geq 1$, let $Q_j  \subseteq P(G)$ be the set of points resulting from the union of $ P_{i} \setminus (P_{i-1} \cup P_{i+1})$ for $i \leq j$,
i.e., $Q_j = \bigcup_{i \leq j} (P_{i} \setminus (P_{i-1} \cup P_{i+1}))$.
We will show that $Q_s$ is a path of length at least $z\delta$.
Then it follows that $G$ contain a path of length at least $z\delta$, according to \autoref{lemma:path:from:PG:to:G}.
First, let us study the intersection of the described paths.

\begin{claim}
\label{claim:long:path:disjoint}
For $i,j \in \{1,\dots,s\}$ and $i<j$ we have
	\begin{itemize}
	\item (A) the interior of $P_i$ and $P_j$ are disjoint,
	unless $j=i+1$ and $\change(p_i)=-1$; and
	\item (B) path $P_{i} \setminus P_{i-1}$ has length $\geq \frac{\delta}{2}$.
	Symmetrically, $P_i \setminus P_{i+1}$ has length $\geq \frac{\delta}{2}$.
\end{itemize}
\end{claim}
Above (A), (B) and that $p_0,\dots,p_s$ are pairwise distinct imply that $Q_s$ is a path in $P(G)$.
\begin{claimproof}
(A)
Let $I$ be the intersection of interior points of $P_i$ and the interior points of $P_j$.
Assume, for the sake of contradiction, that $I \neq \emptyset$.

Consider that $j=i+1$ and $\change(p_i)=1$, meaning $\ori{p_{i-1}}{p_i} \neq \ori{p_{i+1}}{p_i}$.
Let $p_i = p(u,v,\lambda)$ for some $\lambda \in (0,1)$.
By the definition of $\ori{p_{i-1}}{p_i}$ and $\ori{p_{i-1}}{p_i}$,
	the paths $P_i$ and $P_{i+1}$ and $\{ p(u,v,\lambda') \mid \lambda' \in (0,1) \}$
	intersect only in $p_i$.
Then there exists a point $p \in I$ with minimum distance to $p_i$.
Let $P_{i}^\star$ be the subpath of $P_{i}$ from $p_i$ to $p$.
Analogously, $P_{i+1}^\star$ be the subpath of $P_{i+1}$ from $p_i$ to $p$.
Then, $P_i^\star$ is a shortest $p,p_i$-path.
Analogously, $P_{i+1}^\star$ is a shortest $p,p_{i+1}$-path.
We claim that $P_{i}^\star \cup P_{i+1}^\star$ forms a cycle.
If not, there is a point $p' \in (P_{i}^\star \cap P_{i+1}^\star) \setminus \{p,p_i\}$, where $d(p',p_i) < d(p,p_i)$.
Then $p' \in I$ has smaller distance to $p_i$ than $p$, in contradiction to the definition of $p$.
Further, we claim that $p$ is at a vertex.
Assuming otherwise, $p$ is at the position $p = p(u',v',\lambda')$ for some $\lambda' \in (0,1)$.
Then $p(u',v',0) \in I$ or $p(u',v',1) \in I$ and hence $p$ cannot have minimum distance to $p_i$. 

Since $p$ is integral, paths $P_i^\star$, $P_{i+1}^\star$ have equal length and $P_{i}^\star \cup P_{i+1}^\star$ forms a cycle,
	point $p_i$ must also be half-integral.
A contradiction to the definition of a generalized spine.

\medskip

It remains to consider that $j\neq i+1$.
We distinguish between possible lengths of paths $P_i,P_j$, which must be $\delta$ or $\tfrac{\delta}{2}$.
Fix point $p \in I$ arbitrarily.
All cases lead to a contradiction:
\begin{itemize}
\item
Case, paths $P_i$ and $P_j$ both have length $\tfrac{\delta}{2}$.
Then $p_i$, $p_{j-1} \in S$ such that they have distance $\delta$.
However, $d(p_i,p_{j-1}) \leq d(p_i,p) +d(p,p_{j-1}) < \tfrac{\delta}{2} + \tfrac{\delta}{2}$.
Contradiction.
\item
Case, paths $P_i$ and $P_j$ have length $\delta$.
Then $p_i,p_{j-1},p_{j-1},p_j \in S$ such that they pairwise have distance at least $\delta$.
Further $d(p_{i-1},p_i)=\delta$ and $d(p_{j-1},p_j)=\delta$.
It follows that the distance of respectively $p_i,p_{j-1},p_{j-1},p_j$ to $p$ is exactly $\frac{\delta}{2}$.
Then however, $p_{i-1},p_i$ witness the pivot $p$, and hence they may not be part of a spine.
Contradiction.
\item
$\{P_i,P_j\} = \{P',P''\}$ where $P'$ has length $\delta$ and $P''$ has length $\tfrac{\delta}{2}$.
Path $P'$ connects some points $p',p'' \in S$.
Then $d(p',p) \leq \tfrac{\delta}{2}$.
Path $P''$ connects a pivot and some point $q \in S$.
Then $d(q,p) < \tfrac{\delta}{2}$.
Thus, $d(q,p') \leq d(d',p) + d(p,p') < \delta$, in contradiction to that $S$ is $\delta$-dispersed.
\end{itemize}

\medskip

(B)
We show that $P_{i} \setminus P_{i-1}$ has length $\geq \frac{\delta}{2}$, for $i \in \{1,\dots,s\}$.
Then that $P_{i} \setminus P_{i+1}$ has length $\geq \frac{\delta}{2}$, for $i \in \{0,\dots,s-1\}$,
	follows from considering the reversed generalized spine $p_s,p_{s-1},\dots,p_0$.

If $\change(p_{i-1}) = 1$, then (A) applies and we are done.
Thus consider $\change(p_{i-1})=-1$.

If $p_{i-2}$ exists and $p_{i-2} \in \pivots{S}{\delta}$, then $p_{i-2} = p_0$ is a pivot.
Hence $\vel(p_1) = \tfrac{1}{2}$.
Having $\change(p_1) = -1$ implies that $\vel(p_2) = -\tfrac{1}{2}$.
Then points $p_1,p_2$ contradict \autoref{lemma:no:pivot:potentials}.

If $p_i$ exists and $p_{i} \in \pivots{S}{\delta}$, then $p_{i}$ is a pivot and hence $p_i = p_{s}$.
Having that $\change(p_{i-1}) = -1$ it follows that $\vel(p_{i-1})> \vel(p_i) = z$, which we excluded earlier.

If $p_{i-1}$ exists and $p_{i-1} \in \pivots{S}{\delta}$, then $p_{i-1}$ is a pivot and hence $p_0 = p_{i-1}$.
Especially $p_{i-2}$ does not exist, such that (B) follows immediately for $P_{i} \setminus P_{i-1}$.

Thus it remains to consider that points $p_{i-2},p_{i-1},p_{i}$, if they exist, are in $S$.
If paths $P_{i}$ and $P_{i-1}$ have inner points disjoint, we are done.
Hence consider that they intersect in a point $p \notin \{p_{i-2},p_{i-1},p_{i}\}$.
Since $p_{i-2}\neq p_{i}$, there is a well-defined intersection point $p$ that maximizes the distance to $p_{i-1}$.
If $d(p,p_{i-2}) \geq \tfrac{\delta}{2}$, then also $P_{i} \setminus P_{i-1}$ has length $\geq \tfrac{\delta}{2}$.
Thus consider that $d(p,p_{i-2}) < \tfrac{\delta}{2}$.
Then $d(p,p_{i-1}) > \tfrac{\delta}{2}$.
Further $d(p,p_{i+1}) < \tfrac{\delta}{2}$ since $p \in P_{i}$. 
Then we have the contradiction $d(p_{i-2},p_{i}) \leq d(p_{i-2},p) + d(p,p_{i}) < \delta$.
\end{claimproof}

\medskip

Let index $i_{x}$ describe the first point that reaches a velocity $x \in \{1,\dots,z\}$, which is $i_x \coloneqq \min\{ i_x \mid \vel(p_{i_x})\geq x \}$.
We prove by induction on $x \in \{1,\dots,\vel(p_s)-1\}$ that the path $Q_{i_{x+1}-1}$ has length $x \delta$.
(We use that $x < \vel(p_s)$, implies that the point of index $i_{x+1}$ exists.)
For $x = \vel(p_s)$ we have $i_x=s$ and we will follow that $Q_s = s \delta$.
Recall that the difference in the velocity of consecutive points is always $\pm 1$
	with exceptions at the very beginning and at the very end.

Induction base, $x=1$:
Then $i_x = 1$ and $Q_1 = P_1$ has length $\vel(p_1) \in \{\tfrac{1}{2},1\}$, as desired.

Induction step, $x-1 \leadsto x$:
By the induction hypothesis, path $Q_{i_{x}-1}$ has length at least $\delta(x-1)$.
By the definition of $i_x$, the velocity increased twice before reaching $i_x$,
	in other words $\vel(p_{i_x-1})=x-1$ and $\vel(p_{i_x-2})=x-2$.
That means $\change(p_{i_x-1}) = 1$.
Thus the path $Q_{i_{x}-1}$ ending at $p_{i_x-1}$ has inner points disjoint from $P_j$ for $j\geq i_x$,
	according to \autoref{claim:long:path:disjoint}.
If $i_x=s$, then path $Q_{s}$ consists of the disjoint parts $Q_{s-1}$ and $P_{s}$,
	and hence has length $(\vel(p_{s-1}) + (\vel(p_{s})-\vel(p_{s-1})) \delta = \vel(p_s) \delta$,
	as desired.
Else, if $i_x<s$, consider the point with index ${i_{x+1}-1}$.
It remains to show that $Q_{i_{x+1}-1} \setminus Q_{i_{x}-1}$ has length at least $\delta$.

If $i_{x+1}-1 = i_x$, then $Q_{i_{x+1}-1} \setminus Q_{i_{x}-1}$
	consists only of $P_{i_x}$ which has length $\delta$, as desired.

Else, consider that $i_{{x+1}}-1 \neq i_x$.
Then $P_{i_x} \setminus \bigcup_{j \neq i_x} P_j$ has length $\tfrac{\delta}{2}$, according to \autoref{claim:long:path:disjoint}.
Again the velocity increased twice before reaching $i_{x+1}$,
	and hence $\change(p_{i_{x+1}-1}) = 1$.
That means $P_{i_{x+1}-1} \setminus \bigcup_{j \neq i_{x+1}-1} P_j$ has length $\tfrac{\delta}{2}$.
Hence $Q_{i_{x+1}-1} \setminus Q_{i_{x+1}-1}$ has length at least $2 \tfrac{\delta}{2} = \delta$, as desired.

As discussed within the induction step, for $i_x=s$ we follow that $Q_s$ has length $z \delta$.
It follows that $G$ has a path of length $\lfloor z \delta - 1 \rfloor$ as a non-induced subgraph, according to \autoref{lemma:path:from:PG:to:G}.
\end{proof}

\begin{corollary}
	\label{lemma:no:high:velocity}
	Let $P = (p_0,p_1,\dots,p_{s})$ be a generalized spine. 
	Let $\delta = \tfrac{a}{b}$ be a rational number.
	Then $|\vel(p_s)| < \tfrac{b}{2}$.
\end{corollary}
\begin{proof}
Assume, for the sake of contradiction, that $|\vel(p_s)| \geq \tfrac{b}{2}$.
Then by \autoref{lemma:oriented:path} there is a path in $G$ of length $\lfloor |\vel(p_s)|\delta -1 \rfloor \geq \lfloor \tfrac{b}{2} \cdot \tfrac{a}{b} - 1 \rfloor = \lfloor \tfrac{a}{2} - 1 \rfloor$.
Hence $a \leq 2L + 2$.
A contradiction to that $\tfrac{a^\star}{b^\star} > \delta$ where $\tfrac{a^\star}{b^\star}$ is smallest possible rational $\tfrac{a^\star}{b^\star} \geq \delta$ for co-prime $a^\star,b^\star$ and $a^\star < 2L + 2$.
\end{proof}

Finally, we are able to prove \autoref{lemma:equal:potential}.

\begin{lemma}[(\autoref{lemma:equal:potential} restated)]
\label{lemma:equal:potentail:repeat}
\textLemmaEqualPotential
\end{lemma}
\begin{proof}
\newcommand{\bb}{\beta}
By the definition of the root points $R$,
	points $p_0,q_0$ are either both half-integral or $p_0 = q_0$,
	since $p_0,q_0$ must be from the same connected component of $G_S$.
First, let us rule out that $p_0$ is integral and $q_0$ is at the middle position of an edge
	(and analogously the symmetric case).
Consider the walk $p_0,p_1,\dots,p_{i-1},q_j,q_{j-1},\dots,q_0$ in $G_S$.
Remove any loop from this walk to obtain a path $P^\star = (p_0^\star,p_1^\star,\dots,p_{k}^\star)$
	with $p_0^\star = p_0 \neq p_k^\star = q_0$.
Then $P^\star$ is a generalized spine.
\autoref{lemma:lambda:jumps} applies with $\lambda=-\tfrac{1}{2}$.
Therefore $\sgn(p_i) \cdot \lf{\lambda}{\vel(p_i)}$ is the edge position of $q_0$ (relatively to the middle $\tfrac{1}{2}$) which must be $0$.
Hence $\lf{\lambda}{\vel(p_i)} = \fracp(\vel(q_0) \delta) - \tfrac{1}{2} = 0$.
Then $\delta$ must be a rational $\tfrac{a}{b}$ with co-prime $a,b \in \N$
	and $|\vel_{P^\star}(q_0)|$ is a multiple of $\tfrac{b}{2}$,
	which contradicts \autoref{lemma:no:high:velocity}.

\medskip

Thus it remains to consider that $p_0,q_0$ are equal, are both integral or both are at the middle of an edge.
Let $\lambda_0 \in [-\frac{1}{2},\frac{1}{2})$ be such that $p_0= p(\cdot, \orie{p_{1}}{p_0} , \frac{1}{2}+\lambda_0)$.
Analogously, let $\lambda_0' \in [-\frac{1}{2},\frac{1}{2})$ be such that $p_0= p(\cdot, \orie{q_{1}}{q_0} , \frac{1}{2}+\lambda_0')$.
We observe that there is a value $\beta_0 \in \{-1,1\}$ such that $\lambda_0' = \beta_0 \lambda_0$.
If $\lambda_0'= \lambda_0 = 0$, we fix $\beta_0$ to be $1$.

Since $p_i=q_j$, their edge positions defined relatively to their spine $P$ and $Q$ are equal up to the sign.
Also recall that $p_i=q_j$ is not half-integral.
That means there is an edge position $\lambda \in (-\tfrac{1}{2}, \tfrac{1}{2}) \setminus \{0\}$ and a value $\bb' \in \{-1,1\}$
	such that
\begin{align*}\label{eq:claim:equal:potential}
	p_i \; = \; p \big(\cdot, \nori{p_{i-1}}{p_i}, \tfrac{1}{2} + \lambda \big) \; = \; p \big( \cdot,\nori{q_{j-1}}{q_j}, \tfrac{1}{2} + \bb' \lambda \big) \; = \; q_j .
\end{align*}
Figuratively speaking, $\bb'$ indicates whether spines $P,Q$ reach $p_i=q_j$ from the same direction.
We observe that $\nori{p_{i-1}}{p_i} = \nori{q_{j-1}}{q_j}$ if and only if $\beta' = 1$.
Hence for the claim~(2) we have to show that $\beta \coloneqq \beta'\sgn_P(p_i) \sgn_Q(q_j) = 1$.

Let $x = \vel_P(p_i)$ and $y = \vel_Q(q_j)$.
For the claim~(1) we have to show that $x=y$.
By \autoref{lemma:lambda:jumps} we have
	$$ \lambda = \sgn_P(p_i) \lf{\lambda_0}{x} \; \;\text{ and }\; \; \bb'\lambda =  \sgn_Q(q_j) \lf{\beta_0\lambda_0}{y} .$$
Then $\lf{\beta_0 \lambda_0}{y} = \bb \lf{\lambda_0}{x}$ for above defined $\bb \in \{-1,1\}$.
Hence $\lf{\lambda_0}{x} = \bb \lf{\beta_0 \lambda_0}{y}$, which is equivalent to
	$$ \fracp \left( \tfrac{1}{2} + \lambda_0 + x \delta \right) - \tfrac{1}{2}
	= \bb \big( \fracp( \tfrac{1}{2} + \bb_0 \lambda_0 + y \delta ) - \tfrac{1}{2} \big) .$$

We distinguish the cases whether $\bb_0,\bb \in \{-1,1\}$.
\begin{itemize}
\item Case $\bb_0 = \bb = 1$:
Then claim~(2) holds.
Further $ \fracp( \tfrac{1}{2} + \lambda_0 + x \delta )
	= \fracp( \tfrac{1}{2} + \lambda_0 + y \delta ) $,
	which implies that $\fracp(x \delta) = \fracp(y \delta)$.
Assume, for the sake of contradiction, claim~(1) is false, hence that $x \neq y$.
Then $\delta$ is a rational $\frac{a}{b}$ with co-prime $a,b \in \mathbb{N}$.
By symmetry, we may assume that $|x| < |y|$.
Also $x,y$ are both either non-negative or non-positive.
Then $(y-x)\delta \in \mathbb{N}$, and hence $y-x$ is a multiple of $b$.
Then $|y| \geq \tfrac{b}{2}$ in contradiction to \autoref{lemma:no:high:velocity}.
\end{itemize}
All other cases lead to a contradiction, regardless whether $x\neq y$.
\begin{itemize}
		\item
		Case $\bb_0 = \bb = -1$:
		Then
		$ \fracp( \tfrac{1}{2} + \lambda_0 + x \delta ) - \tfrac{1}{2}
			= -\fracp( \tfrac{1}{2} - \lambda_0 + y \delta ) + \tfrac{1}{2} $.
			
		Thus $ \fracp( \tfrac{1}{2} + \lambda_0 + x \delta )  +\fracp( \tfrac{1}{2} - \lambda_0 + y \delta ) = 0  $.
		
		Hence $ (\tfrac{1}{2} + \lambda_0 + x \delta) + (\tfrac{1}{2} - \lambda_0 + y \delta ) = (x+y)\delta \in \mathbb{Z}  $.
		
		This implies that $\delta$ is a rational $\frac{a}{b}$ with co-prime $a,b \in \mathbb{N}$,
			and that $|x+y|$ is a multiple of $b$.
For either $z \in \{x,y\}$ we have $|z| \geq \frac{|x + y|}{2}$.
By symmetry, assume that $|x| \geq \frac{|x+y|}{2}$.
Thus $|x|$ is larger than a multiple of $\tfrac{b}{2}$,
	in contradiction to \autoref{lemma:no:high:velocity}
\item
Case $\bb=1$ and $\bb_0 = -1$:
Then $\fracp(\tfrac{1}{2} + \lambda_0 + x\delta) = \fracp(\tfrac{1}{2} - \lambda_0 + y\delta)$.

Thus $\fracp(\tfrac{1}{2} + \lambda_0 + x\delta) - \fracp(\tfrac{1}{2} - \lambda_0 + y\delta) = 0$.
		
That means $(\tfrac{1}{2} + \lambda_0 + x\delta) - (\tfrac{1}{2} - \lambda_0 + y\delta) \in\mathbb{Z}$ and hence $2\lambda_0 + (x-y)\delta \in \mathbb{Z}$.
		
Let $z \coloneqq \tfrac{x-y}{2}$.
Then  $\lf{\lambda_0}{z} = \fracp(\tfrac{1}{2} (1+2\lambda_0 + 2z\delta)) - \tfrac{1}{2} \in \{0,-\tfrac{1}{2}\}$, i.e., it is half-integral.
Hence no point $p \in \{p_1,\dots,p_{i}\}$ has velocity $\lf{\lambda_0}{ \vel_P(p) } = z$, since $p$ is not half-integral.
Thus $z \notin \{x,y\}$.

Further $z\neq 0$;
	Indeed assuming that $z=0$, implies that
	$2 \lambda_0 + 0\delta \in \mathbb{Z}$ and hence that $\lambda_0$ is half-integral.
	In that case, however, $\bb_0 = 1$, and hence this is covered by other cases.

Now, by symmetry, we may assume that $z > 0$.
Either $z$ occurs as a potential $\vel_P(p_{i'})$ or
	$z-\tfrac{1}{2},z+\tfrac{1}{2}$ occur as potentials $\vel_P(p_{i'}), \vel_P(p_{i+1'})$ for some $j \in \{1,\dots,i-1\}$.
The former implies that $p_{i'}$ is half-integral.
The latter implies that $p_{i'-1}, p_{i'}$ witness a pivot half-way on a shortest path between them.
Both cases contradict the definition of the spine $P$.
		
\item
Case $\bb = -1$ and $\bb_0 = 1$:
Then $ \fracp( \tfrac{1}{2} + \lambda_0 + x \delta ) + \fracp( \tfrac{1}{2} + \lambda_0 + y \delta ) = 1  $.

That means $ \tfrac{1}{2} + \lambda_0 + x \delta + \tfrac{1}{2} + \lambda_0 + y \delta \in \mathbb{Z}$,
	and hence $2\lambda_0 + (x+y) \delta  \in \mathbb{Z}$.

Let $z=\tfrac{x+y}{2}$.
Then $\lf{\lambda_0}{z} = \fracp(\tfrac{1}{2} (1 + 2\lambda_0 + (x+y)\delta) \in \{0,-\tfrac{1}{2}\}$, i.e., it is half-integral.
Hence no point $p \in \{p_1,\dots,p_{i}\}$ has velocity $\lf{\lambda_0}{ \vel_P(p) } = z$, since $p$ is not half-integral.
Note that $z \notin \{0,x,y\}$ since $x,y\geq 1$.
Then, analogously to the previous case it follows either contradiction that an inner point of $P,Q$ is half-integral or two points witness a pivot.
\end{itemize}
To summarize, we always have $\beta = \beta_0 = 1$ and there claim~(1) and~(2) hold.
\end{proof}

\section{Hardness Results}

\subsection{Diameter}

\begin{lemma}[(\autoref{lemma:hardness:chordal} restated)]
\label{lemma:appendix:hardness:chordal}
\lemmaTextHardnessChordal
\end{lemma}

\textbf{Construction:}
Let $\delta > 3$ be arbitrary.
Given a \problemm{Independent Set}-Instance consisting of a graph $G$ and integer $k$,
	we construct a $\delta$-\dispersion-instance
	consisting of a graph $G'$ and integer $k'=k$.
For every edge $\{u,v\} \in E(G)$, add a vertex $w_{\{u,v\}}$ to $V(G')$.
Let $w_{\{u,v\}}$ for $\{u,v\} \in E(G)$ form a clique.
For every vertex $u \in V(G)$ add a vertex $u'$ to $V(G')$
	and make it adjacent to $w_{\{u,v\}}$ for every $v \in N_G(u)$.
For every $u \in V(G)$, add a path of length $\left\lceil\tfrac{\delta}{2}\right\rceil-2 \geq 0$ starting at $u'$ and terminating at a vertex which we call $u_1$.
Finally, if $\left\lceil{\delta}\right\rceil$ is even,
	add a true twin $u_2$ to $u_1$ for every $u\in V(G)$,
	that is $u_2$ is adjacent to $\{u_1\} \cup N_{G'}(u_1)$.

\medskip

The diameter of $G'$ is at most the distance between $u_1$ and $v_1$ for non-adjacent $u,v \in V(G)$,
	hence at most $2(\left\lceil\tfrac{\delta}{2}\right\rceil-2)+3 \leq \delta$.
Observe that $G'$ is connected and chordal.
Further, the construction is possible in polynomial time.
Moreover, \problemm{Independent Set} is \wone-hard and the reduction preserves the parameter.
Hence, in order to show \np-hardness and \wone-hardness,
	it remains to show the correctness:

\begin{lemma}
Graph $G$ has an independent set of size $k$
	if and only if $G'$ has a $\delta$-dispersed set of size $k$.
\end{lemma}
\begin{proof}
\forward
Let $I \subseteq V(G)$ be an independent set of $G$.
If $\left\lceil{\delta}\right\rceil$ is odd,
	then let $S\subseteq P(G)$ be the point set consisting of points $p_u$ at the vertices $u \in I$.
Then, every $p_u,p_v \in S$ have distance $2(\left\lceil\tfrac{\delta}{2}\right\rceil-2)+ 3 \leq \delta$, since $\{u,v\} \notin E(G)$.
Else, when $\left\lceil{\delta}\right\rceil$ is even,
	let $S \coloneqq \{ p(u_1,u_2,\tfrac{1}{2}) \mid u \in I\}$.
Then, every $p(u_1,u_2,\tfrac{1}{2})$, $p(v_1,v_2,\tfrac{1}{2})$
	have distance $2(\left\lceil\tfrac{\delta}{2}\right\rceil-2)+4 \leq \delta$, again since $\{u,v\} \notin E(G)$.
Thus, in both cases, $S$ is $\delta$-dispersed.

\backward
Let $S \subseteq P(G')$ be a $\delta$-dispersed set of $G'$.
Let $C \subseteq P(G')$ be the set of points located at the central clique,
	that is the set of points $p(u,v,\lambda)$ with $u,v \in V(G)$ and $\lambda \in  [0,1]$.
Consider that $S$ contains a point $p \in C$,
	hence a point $p=p(u,v,\lambda)$ with $u,v \in V(G)$ and $\lambda \in  [0,1]$.
Then $p$ has distance to any vertex $u'$ for $u \in V(G)$ of at most $\tfrac{3}{2}$.
Thus, $p$ has distance to any other point in $P(G)$ of at most $\tfrac{3}{2} + \left\lceil\tfrac{\delta}{2}\right\rceil-2< \delta$.
That means, $|S|=1$ and any vertex in $G$ forms an independent set of the same size.
Hence it remains to consider that $S \cap C = \emptyset$.

Let \define{ball} $B_u \subseteq P(G)$, for $u \in V(G)$, consist of every point in distance $\left\lceil\tfrac{\delta}{2}\right\rceil-1\geq 1$ to $u_1$ or $u_2$ (if $u_2$ exists).
Then the balls $B_u$ for $u \in V(G)$ and $C$ together make up the whole set of points $P(G)$.
Clearly, every ball $B_u$, for $u \in V(G)$, contains at most one point form $S$.
Then $I \coloneqq \{ u \in V(G) \mid S \cap B_u \neq \emptyset\}$ has size $|S|$.
We claim that $I$ is an independent set.
First consider that $\delta >4$.
Consider an edge $\{u,v\} \in E(G)$.
Vertices $u'$ and $v'$ have the common neighbor $w_{u,v}$,
	hence have distance $2$.
Any point $p \in B_u$, for $u \in V(G)$, has distance at most $\left\lceil\tfrac{\delta}{2}\right\rceil-\tfrac{3}{2}$ to $u'$.
Hence, any distinct points $p,q \in B_u \cup B_v$ have distance at most $2 (\left\lceil\tfrac{\delta}{2}\right\rceil-\tfrac{3}{2})+1 < \delta$.
Therefore, for adjacent $u,v \in V(G)$ either $S \cap B_u \neq \emptyset$ or $S \cap B_v \neq \emptyset$, such that $|I \cap \{u,v\}|=1$,
	hence that $I$ is an independent set of $G$.

A remaining special case is when $\delta \leq 4 $.
Then $u',v'$ have a true twins $u_2$ and $v_2$ respectively,
	the distance of any $u',u_2$ to any $v',v_2$ is at most $2$.
Again, any point $p \in B_u$ has distance at most $\left\lceil\tfrac{\delta}{2}\right\rceil-\tfrac{3}{2}$ to either $u'=u_1$ or its true twin $u_2$.
Thus, analogously follows that $I$ is an independent set of $G$.
\end{proof}

\subsection{Pathwidth and Feedback Vertex Set Size}

\begin{theorem}[(\autoref{lemma:pw:fvs:eth} restated)]
\label{lemma:appendix:pw:fvs:eth}
\lemmaTextPwFvsETH
\end{theorem}

In \problemm{Multi-Colored-Independent-Set} (\problemm{MCIS}) we are given a graph $G$ with a partition of $V(G)$ into 
$k$ independent sets, called \define{color classes}, $V_1,\dots,V_k$, each of size $n$.
The task is to find a \define{multi-colored independent set},
	an independent set that contains exactly one vertex from each color class.

\textbf{Construction (as in \cite{DBLP:conf/wg/KatsikarelisLP18})}
Let the given \problemm{MCIS} consist of graph $G$ where $V(G)$ partitions into $k$ color classes $V_1,\dots,V_k$, each of size $n$.
We construct a \dispersion instance consisting of a graph $G'$, distance $\delta \coloneqq 6n$
	that asks for a set of size at least $k' \coloneqq k^2$.
\begin{itemize}
\item
For every color class $V_i$, $1 \leq i\leq k$:
	Enumerate the vertices in $V_i$ as $v_1^i,\dots, v_n^i$.
	Create a set $P_i \subseteq V(G')$ of $n$ vertices $p_1^i,\dots, p_n^i$;
	Add vertex a $a_i$ and connect it to each $p_\ell^i$ by a path of length $n+l$ for $1 \leq l \leq n$;
	Similarly add a vertex $b_i$ and connect it to each  $p_\ell^i$ by a path of length $n+l$ for $1 \leq l \leq n$.
\item
For every non-edge $e = \{ v_{j_1}, v_{j_2} \} \in V_{i_1} \times V_{i_2} \setminus E(G)$ (where ${i_1} \neq {i_2}$) add a vertex $u_e$.
We connect $u_e$ to $a_{i_1},b_{i_1}$ and to $a_{i_2},b_{i_2}$ by paths of lengths that encode the indices $j_1$ and $j_2$:
Add a path between $u_e$ and $a_{i_1}$ of length $5n - j_1$ and add a path between $u_e$ and $b_{i_1}$ of length $4n+ j_1$;
Similarly between $u_e$ and $a_{i_2}$ of length $5n - j_2$ and add a path between $u_e$ and $b_{i_2}$ of length $4n+ j_2$.
\item
For every distinct indices $i_1,i_2 \in \{1,\dots,k\}$ introduce a path of length $6n-1$ with start vertex $g_{i_1,i_2}$ and end vertex $g_{i_1,i_2}'$.
Make $g_{i_1,i_2}$ adjacent to every $u_e$ with $e \in V_{i_1} \times V_{i_2}$.
\end{itemize}

\medskip

A multi-colored independent set $I \subseteq V(G)$ translates to a  $6n$-dispersed set of size $k^2$
	by selecting the corresponding $p^{i}_j$ vertices and
	for every pair of selected $p^{i_1}_{j_1}, p^{i_1}_{j_2}$, for $i_1\neq i_2$, selecting $g_{i_1,i_2}'$ and $u_{e}$
	where $e = \{v^{i_1}_{j_1}, v^{i_1}_{j_2}\}$.
Similarly, a maximum $6n$-dispersed is forced to select vertices corresponding to an independent set and non-edges between those in the original graph $G$.
The details of the correctness easily follow from the reduction of that of Katsikarelis et al.~\cite{DBLP:conf/wg/KatsikarelisLP18}.

The resulting graph $G'$ has a feedback vertex set of size $\Oh(k)$
	which consists of $a_i,b_i$ for every $1\leq i \leq k$.
Also $G'$ has pathwidth $\Oh(k^2)$:
Removing all $a_i, b_i$ vertices and all $g_{i_1,i_2}$ vertices results in a forest of stars
	which has a constant width pathwidth decomposition.
By adding the $\Oh(k^2)$ removed vertices to every bag,
	we obtain a pathwidth decomposition of $G'$ and of width $\Oh(k^2)$ .

The bounds on the feedback vertex set size and pathwidth
	imply lower bounds assuming the Exponential Time Hypothesis (ETH).
\problemm{MCIS} cannot be solved in time $n^{o(k)}$ assuming ETH is true \cite{DBLP:books/sp/CyganFKLMPPS15}.
However, assuming an $n^{o(\sqrt{\pw(G)}+ \sqrt{k})}$-time algorithm for \dispersion
	implies an $n^{o(k)}$-time algorithm for \problemm{MCIS}.
Similarly, an $n^{o({\fvs(G)}+ \sqrt{k})}$-time algorithm for \dispersion
	implies an $n^{o(k)}$-time algorithm for \problemm{MCIS}.

\subsection{Treedepth}

\begin{theorem}[(\autoref{lemma:td:lower:bound} repeated)]
\label{lemma:appendix:td:lower:bound}
\lemmaTextTDlower
\end{theorem}
To show a lower bound regarding the treedepth,
	a start is to use same reduction as for parameter pathwidth.
However, this alone provides only a lower bound of $n^{\Oh(k + \log \ell)}$
	where $\ell$ is the number of colors of the original \problemm{MCIS}-instance.
To avoid the dependency on $n$ we further use a reduction from $3$-\problemm{SAT} to \problemm{MCIS} that outputs instance with relatively small $\ell$
(similar as in \cite{DBLP:conf/wg/KatsikarelisLP18}).

First, consider the treedepth of $G'$ of the resulting \dispersion-instance
	when the original \problemm{MCIS}-instance consists of a graph $G$ with $k$ color classes each containing $\ell$ vertices.
The treedepth of $G'$ is at most $\Oh(k + \log \ell)$,
	as observed in \cite{DBLP:conf/wg/KatsikarelisLP18}:
To see this, consider the components of graph $G'$ after removing the $2k$ vertices $a_i,b_i$ for $1 \leq i \leq k$.
The treedepth of $G'$ is then at most $2k$ plus the maximum treedepth of the remaining components.
The remaining components are a tree for each vertex $g_{i_1,i_2}$ 
	and a path of length $3\ell$ for each vertex $p^i_j$.
A path of length $3\ell$ has treedepth $\Oh(\log \ell)$.
For a tree containing some $g_{i_1,i_2}$,
	note that after removing $g_{i_1,i_2}$ only paths of length $<6 \ell$ remain.
Hence, each remaining tree has treedepth $\Oh(\log \ell)$.
In total $G'$ has treedepth $\Oh(k + \log \ell)$.

Now, we give a (non-polynomial time) reduction from $3$-\problemm{SAT} to \problemm{MCIS}.
There, a $3$-\problemm{SAT}-instance $\varphi$ with $N$ variables and $\Oh(N)$ clauses reduces to a \problemm{MCIS}-instance with $\sqrt{N}$ color classes each containing $c^{\sqrt{N}}$ vertices,
	for some constant $c$.
Further, the run-time is $\Oh(\sqrt{N} \cdot c^{2\sqrt{N}} )$.

\medskip

\textbf{Construction (similar to \cite{DBLP:conf/wg/KatsikarelisLP18}\footnote{The reduction there directly reduces to \problemm{Distance Independent Set}.
	 The idea is essentially to use the following reduction to \problemm{MCIS} as the first step.}):}
Group the clauses of $\varphi$ into $\ell = \Oh(\sqrt{N})$ groups $F_1,\dots,F_{\ell}$ each of size $\Oh(\sqrt{N})$.
Thus, every group $F_i$ consists of $\Oh(\sqrt{N})$ variables,
	which have $2^{\Oh(\sqrt{N})}$ truth assignments.
Let $c$ be such that $2^{\Oh(\sqrt{N})} = c^{\sqrt{N}}$.
Let $\varphi^i_1,\dots,\varphi^i_{s}$ for $s \leq c^{\sqrt{N}}$ be an enumeration of truth assignments to the variables of $F_i$ that satisfy the clauses of $F_i$, for each $i \in \ell$.
For each truth assignments $\varphi^i_j$ introduce a corresponding vertex $u^i_j$ of color $i$.
Add an edge between $u^i_j$ and $u^{i'}_{j'}$ for $i \neq i'$, if the truth assignments $\varphi^i_j$ and $\varphi^{i'}_{j'}$ are contradicting each other in variables they have in common.

\medskip

It easily follows that $\varphi$ is satisfiable if and only if the constructed graph $G'$ has a multi-colored independent set of size $\ell$.
The run-time is $\Oh(\sqrt{N} \cdot c^{\sqrt{N}} + c^{2\sqrt{N}} )$.
For more details, we refer to~\cite{DBLP:conf/wg/KatsikarelisLP18}.

\medskip

By combining these two reductions, $3$-\problemm{SAT} reduces to \dispersion with treedepth $\Oh(k + \log \ell) = \Oh( \sqrt{N} + \log c^{\sqrt{N}} ) = \Oh(\sqrt{N})$.
Assuming a $2^{o(\td(G)^2)}$-time algorithm for \dispersion implies a $2^{o(N)}$-time algorithm for $3$-\problemm{SAT} with $N$ variables and $\Oh(N)$ clauses,
	contradicting ETH~\cite{DBLP:journals/jcss/ImpagliazzoPZ01}.
This finishes the proof of \autoref{lemma:appendix:td:lower:bound}.

\end{document}